%% file: JoDSPaper.tex
\documentclass[twocolumn,envcountsame,final]{svjour3}

\usepackage{rotating}
\usepackage[latin9]{inputenc}
\usepackage[T1]{fontenc}
\usepackage{amsmath}
\usepackage{amssymb}
\usepackage{multirow}

\usepackage{mathptmx} 
\DeclareMathAlphabet{\mathcal}{OMS}{cmsy}{m}{n} 

\usepackage{mathrsfs}
\usepackage{macros}

\usepackage[english]{babel}
\usepackage{color}
\usepackage{booktabs}
\usepackage{graphicx}
\usepackage{hyperref}
\usepackage{boxedminipage}
\usepackage{cite}
\usepackage{url}

\usepackage{xspace}
\usepackage{todonotes}

\smartqed

\definecolor{Gray}{gray}{0.85}

\usepackage{xcolor}
\usepackage{framed}
\usepackage{colortbl}

\newcommand{\exfo}
    [1]{\ensuremath{\mathsf{#1}}}


\usepackage{array}
\newcolumntype{H}{@{}>{\lrbox0}l<{\endlrbox}}

\newcommand{\Onto}{\Omc}
\newcommand{\A}{\Amc}
\newcommand{\T}{\Tmc}
\newcommand{\I}{\Imc}
\newcommand{\J}{\Jmc}
\newcommand{\ALCHN}{\ensuremath{\mathcal{ALCHN}}\xspace}
\newcommand{\SHIN}{\ensuremath{\mathcal{SHIN}}\xspace}
\newcommand{\SHQ}{\ensuremath{\mathcal{SHQ}}\xspace}
\newcommand{\SHIQ}{\ensuremath{\mathcal{SHIQ}}\xspace}

\newcommand{\SRIQ}{\ensuremath{\mathcal{SRIQ}}\xspace}
\newcommand{\SROQ}{\ensuremath{\mathcal{SROQ}}\xspace}
\newcommand{\SROIN}{\ensuremath{\mathcal{SROIN}}\xspace}
\newcommand{\SROIQ}{\ensuremath{\mathcal{SROIQ}}\xspace}
\newcommand{\SROIQBs}{\ensuremath{\mathcal{SROIQB}_\text{s}}\xspace}
\newcommand{\axiom}[1]{\ensuremath{\langle #1\rangle}\xspace}
\newcommand{\atLeast}[3]{\ensuremath{{\ge}#1\,#2.#3}\xspace}
\newcommand{\atMost}[3]{\ensuremath{{\le}#1\,#2.#3}\xspace}
\newcommand{\uatLeast}[2]{\ensuremath{{\ge}#1\,#2}\xspace}
\newcommand{\uatMost}[2]{\ensuremath{{\le}#1\,#2}\xspace}
\newcommand{\Self}{\ensuremath{\text{Self}}\xspace}
\newcommand{\crisp}{\ensuremath{_\mathsf{c}}\xspace}

\newcommand{\godel}{G\"odel\xspace}
\newcommand{\lukas}{\L u\-ka\-sie\-wicz\xspace}
\newcommand{\tnorm}{\ensuremath{\otimes}\xspace}
\newcommand{\tconorm}{\ensuremath{\oplus}\xspace}
\newcommand{\negation}{\ensuremath{\ominus}\xspace}
\newcommand{\implication}{\ensuremath{\Rightarrow}\xspace}

\newcommand{\chain}{\ensuremath{\mathscr{C}}\xspace}
\newcommand{\chainp}{\ensuremath{\mathscr{C}_{>0}}\xspace}
\newcommand{\Var}[1]{\ensuremath{\textsf{VarInds}(#1)}\xspace}
\newcommand{\qt}{{\ensuremath{q_\mathsf{t}}}\xspace}
\newcommand{\qf}{{\ensuremath{q_\mathsf{f}}}\xspace}
\newcommand{\qs}{{\ensuremath{q_\mathsf{s}}}\xspace}
\newcommand{\cq}{\ensuremath{\textnormal{CQ}}\xspace}
\global\long\def\next#1{#1_{\textnormal{next}}}
\newcommand{\lb}{\linebreak}
\newcommand{\NV}{\ensuremath{\mathsf{N_V}}\xspace}

\newcommand{\ans}{\ensuremath{\mathfrak{a}}\xspace}
\newcommand{\TwoExpTime}{\textsc{2\nobreakdash-ExpTime}\xspace}
\newcommand{\ThreeExpTime}{\textsc{3\nobreakdash-ExpTime}\xspace}

\newcommand{\pagoda}{\mbox{PAGOdA}\xspace}
\newcommand{\delorean}{\mbox{DeLorean}\xspace}
\newcommand{\hermit}{\mbox{HermiT}\xspace}

\spdefaulttheorem{exampleb}{Example}{\itshape\color{blue}}{\rmfamily\color{blue}}
\begin{document}

\title{Answering Fuzzy Conjunctive Queries over Finitely Valued Fuzzy
Ontologies}

\author{Stefan Borgwardt \and Theofilos Mailis \and Rafael Pe\~naloza \and Anni-Yasmin Turhan}

\institute{S. Borgwardt \and A.-Y. Turhan \at
		Chair for Automata Theory, Institute for Theoretical Computer Science, Technische Universit\"at Dresden,
		Germany
		\email{\{stefan.borgwardt,anni-yasmin.turhan\}@tu-dresden.de} 
		\and 
		T. Mailis \at Department of Informatics and Telecommunications, 
		National and Kapodistrian University of Athens, Greece
		\email{theofilos@image.ntua.gr}
		\and
		R. Pe\~naloza \at KRDB Research Centre, Free University of Bozen-Bolzano, Italy
		\email{rafael.penaloza@unibz.it}
		}
\maketitle

\input{1.intro}

\input{2.preliminaries}
\input{3.reduction}

\input{5.cqs_answering}

\input{5.implementation}

\input{6.related_work}

\input{7.conclusions}

\section*{Acknowledgements}

This work was partially supported by the 
	German Research Foundation (DFG) under the research grant BA 1122/17-1 (FuzzyDL), 
	the Collaborative Research Center 912 ``Highly Adaptive Energy-Efficient Computing,'' and the Cluster of Excellence ``Center for Advancing Electronics Dresden'';
	it was developed while R.\ Pe\~naloza was affiliated with TU Dresden and the Center for Advancing Electronics Dresden, Germany.

We also want to thank Fernando Bobillo for providing us with a binary of the 
DeLorean system, and the anonymous 
reviewers for their valuable comments on earlier drafts of this paper.

\section*{Appendix}

\input{8.appendix}

\bibliographystyle{spmpsci}
\bibliography{citations}

\end{document}

%% file: 1.intro.tex

\begin{abstract}
  Fuzzy Description Logics (DLs) provide a means for representing vague knowledge
  about an application domain.
  In this paper, we study fuzzy extensions of conjunctive queries (CQs) over
  the DL \SROIQ based on finite chains of degrees of truth. To answer such
  queries, 
  we extend a well-known technique that reduces the fuzzy ontology
  to a classical one, and use classical DL reasoners as a black box.
  We improve the complexity of previous reduction techniques for
  finitely valued fuzzy DLs, which allows us to prove tight complexity results
  for answering certain kinds of fuzzy CQs.
  We conclude with an experimental evaluation of a prototype implementation,
  showing the feasibility of our approach.
\end{abstract}

\section{Introduction}


Description Logics (DLs) are a family of knowledge representation
languages with unambiguous syntax and well-defined semantics that are widely used to
represent the conceptual knowledge of an application domain in a
structured and formally well-understood manner.
DLs have been successfully employed to formulate ontologies for
a range of knowledge domains, in particular for the bio-medical sciences.
Prominent examples of ontologies in these areas are the Gene Ontology,%
\footnote{\url{http://geneontology.org/}}
and the ontology {\sc Snomed\,CT}.%
\footnote{\url{http://ihtsdo.org/snomed-ct/}}
Arguably the largest success of DLs to date is that they provide the
formal foundation for the standard web ontology language
OWL, a milestone for the Semantic Web. 
More precisely, the  current version of the web ontology language, OWL~2,
is based on the very expressive DL \SROIQ.%
\footnote{\url{http://www.w3.org/TR/owl2-overview/}}

In DLs, knowledge is represented through \emph{concepts} that describe collections of objects 
(that is, correspond to unary predicates from first-order logic), and \emph{roles}
that define relations between pairs of objects (binary predicates).
To encode the actual knowledge of the domain, DLs employ different kinds of
\emph{axioms}. These axioms restrict the possible interpretations of the concepts and 
roles.
For example, we can express the fact that \exfo{cpuA} is an
overused CPU, and that every server that has a part that is overused is
a server with limited resources through the axioms
\begin{align}
  &(\exfo{CPU}\sqcap \exfo{Overused})(\exfo{cpuA})
    \label{exa:axiom:first} \\
  &\exfo{Server} \sqcap \exists \exfo{hasPart}. \exfo{Overused} \nonumber\\
  &\qquad\qquad\qquad\sqsubseteq \exfo{ServerWithLimitedResources}
  \label{exa:axiom:second}
\end{align}
An axiom of the form~\eqref{exa:axiom:first} is called an \emph{assertion},
while~\eqref{exa:axiom:second} is a \emph{general concept inclusion} (GCI).

It has been widely argued that many application domains require the representation
of vague concepts, for which it is impossible to precisely characterize the objects
that belong to these concepts, and distinguish them from those who do not belong
to them~\cite{straccia2013foundations}. 
A simple example of such a concept is that of an overused CPU. While it
is easy to state that a CPU that is running permanently at its maximum capacity is
overused, and one that is not being used at all is not overused, there is no
precise usage point where a CPU starts (or stops) being overused. Fuzzy Description
Logics have been proposed to alleviate this problem. In these logics,
objects are assigned a \emph{membership degree}, typically a number between
0 and 1, expressing ``how much'' they belong to a given concept. 
In general, the higher the degree of an object, the more it belongs to the concept.
To represent vague knowledge, axioms are also extended to restrict the possible
degrees that the interpretations may use. 
Thus, one can express that \exfo{cpuA} is overused with degree at least 0.8 through
the assertion \mbox{$\exfo{Overused}(\exfo{cpuA})\geqslant 0.8$}.

Formally, fuzzy DLs generalize classical DLs by interpreting concepts and roles
as fuzzy unary predicates and fuzzy binary predicates, respectively. Hence,
fuzzy DLs can be seen as sublogics of fuzzy first-order logic. Adopting this view,
one can use a triangular norm (\emph{t-norm}) and its associated operators to interpret the different logical 
constructors. Each t-norm then defines a specific family of fuzzy DLs.

It has been shown that reasoning in fuzzy DLs easily becomes undecidable, if
infinitely many membership degrees are allowed~\cite{BaBP-JPL14,BoDP-AI15}.
In fact, these undecidability results hold even for relatively inexpressive fuzzy DLs.
This has motivated the study of \emph{finitely valued} fuzzy DLs.

It is known that the complexity of standard reasoning tasks in expressive DLs
is not affected by the use of t-norm based finitely valued 
semantics~\cite{BoPe-JoDS13,borgwardt2013consistency}. Unfortunately,
the automata-based techniques exploited in~\cite{BoPe-JoDS13} cannot be easily
adapted to obtain complexity bounds for the problem of answering conjunctive
queries in these logics. Moreover, despite providing optimal complexity bounds,
automata-based methods are not used in practice due to their bad best-case
behavior.

A different approach for reasoning in the presence of finitely many membership degrees 
is \emph{crispifying}; i.e., transforming
a finitely valued ontology into an ``equivalent'' classical ontology, from
which the relevant membership degrees can be
read~\cite{BoDG-IJUF09,bobillo2013finite,straccia2004transforming}.
Reasoning in finitely valued fuzzy DLs is thus reduced to reasoning in classical DLs,
for which very efficient methods have already been developed and implemented. The main
drawback of the translation described in~\cite{BoDG-IJUF09,bobillo2013finite}
is that it may introduce an exponential blow-up of the ontology, thus affecting the efficiency
of the overall method.

In this paper, we adapt the crispification approach for answering conjunctive queries
in expressive finitely valued fuzzy DLs. The problem of answering conjunctive queries 
has recently received
much attention as a powerful means to access facts encoded in an ontology. 
For example, using a fuzzy conjunctive query it is possible to ask for all pairs of
servers and CPUs such that the CPU is an overused part (to degree at least $0.6$)
 of the server as follows:
\begin{multline*}
  \{\exfo{Server}(x)\geqslant 1 ,~
  \exfo{hasPart}(x,y)\geqslant 1 ,~
  \exfo{CPU}(y)\geqslant 1 ,{} \\
  \exfo{Overused}(y)\geqslant 0.6\}.
\end{multline*}
The crispification approach allows us to effectively answer conjunctive queries 
over finitely valued ontologies, by re\-using the methods developed for the 
classical case. Once the ontology is crispified, this approach calls a classical conjunctive query answering engine
as a black-box procedure. Thus any optimization developed for the classical
case automatically improves the performance for the finitely valued
scenario. What remains to be addressed is the exponential blow-up of the ontology, if done according to \cite{BoDG-IJUF09,bobillo2013finite}.
We strengthen our results by providing a linear preprocessing
step that avoids the exponential blow-up produced by this
crispification.
Using this preprocessing of the finitely valued ontology, we can
guarantee that the classical ontology produced is only polynomially
larger than the original input. In particular, this means that the
classical query answering engine becomes able to provide answers more
efficiently over classical ontologies of lesser size.

\medskip
\noindent
The contributions of this paper are the following:
\begin{itemize}
\item We prove that some of the previous crispification
  algorithms~\cite{BDGS-IJUF12,bobillo2011reasoning,MaPe-RR-14} are incorrect
  for qualified number restrictions (indicated by the letter~\Qmc in the name
  of DLs) by means of a counter-example (see
  Example~\ref{exa:number-restrictions-incorrect}).
\item We discuss a possible way to reduce such qualified number restrictions,
  but which depends on the presence of so-called Boolean role
  constructors~\cite{RuKH-JELIA08} in the DL (see
  Section~\ref{subsec:number-restrictions}).
\item We improve the reduction from finitely valued \SROIN ontologies to classical \SROIN
  ontologies by introducing a linear normalization step (\Nmc stands for
  \emph{unqualified} number restrictions, see Section~\ref{sec:reduction}).
\item We extend the crispification approach to answering different types of 
fuzzy
  conjunctive queries in the finitely valued setting and we prove 
  correctness of the obtained methods (see Section~\ref{sec: reduction based conjunctive query
    answering}). This approach works for any known crispification
  algorithm, in particular specialized ones that correctly reduce
  number restrictions.
\item We assess the complexity of the presented conjunctive query answering
  technique for a family of fuzzy extensions of (sublogics of) \SROIQ in regard of two types of conjunctive queries that use membership degrees.
\item We provide an evaluation of a prototype implementation of our methods over the LUBM ontology benchmark~\cite{guo2005lubm}
based on the reduction-based 
\delorean reasoner\cite{BoDG-ESA12} for fuzzy ontologies and a standard query 
answering reasoner for crisp ontologies (\pagoda~\cite{ZCGNH-DL-15}).
\end{itemize}

A preliminary version of this paper can be found in~\cite{MaPe-RR-14}, where
the incorrect reduction of number restrictions was still used.
We also extend here that earlier paper by full proofs for the correctness of the crispification procedure, something which has not been done before in the literature.
Finally, we optimize the crispification procedure 
described~\cite{MaPe-RR-14} in order to eliminate an
exponential blow-up inherent in some of the previous crispification proposals.

The rest of the paper is structured as follows:
Section~\ref{sec:preliminaries} introduces the syntax and semantics of
finitely valued fuzzy DLs based on \SROIQ.
Section~\ref{sec:reduction} describes our improved reduction procedure from
fuzzy to classical ontologies.
Section~\ref{sec: reduction based conjunctive query answering} presents the
actual reduction from fuzzy to classical conjunctive query answering.
Section~\ref{sec:results} provides an evaluation of a prototype implementation
over (fuzzified versions of) the LUBM ontology benchmark~\cite{guo2005lubm}.
Finally, Section~\ref{sec: related work} presents the current literature on
reduction techniques and conjunctive query answering for fuzzy DLs and
Section~\ref{sec: conlusions} summarizes the paper and mentions directions for
future work.

%% file: 2.preliminaries.tex

\section{Preliminaries}\label{sec:preliminaries}

We first introduce a class of finite chains, together with some basic operations over them.
Afterwards, we formally define the fuzzy extension of \SROIQ, whose semantics is based 
on these chains.

\subsection{Finite Fuzzy Logics}

The semantics of fuzzy DLs is based on truth structures endowed with
additional operators for interpreting the logical constructors. We consider
arbitrary \emph{finite} total orders (or \emph{chains}). Since the names
of the truth degrees in a chain are not relevant, we 
consider in the following only the canonical chain of $n$ elements
$\chain:=\{0,\frac{1}{n-1},\dots,\frac{n-2}{n-1},1\}$, in the
usual order.
We denote $\chain\setminus\{0\}$ by~\chainp.
We use the notation
$\next{d}$ to refer to the \emph{direct upper neighbour} of~$d$ in~\chain, which
is the unique smallest element strictly larger than~$d$.
We now consider tuples of the form
$(\chain,\tnorm,\implication,\negation,\tconorm)$ that specify the
chain together with the operators used for interpreting conjunction,
implication, negation, and disjunction, respectively.

The largest family of operators used for fuzzy semantics is based on \emph{t-norms}, which are
associative, commutative binary operators that are monotonic in both arguments
and have identity $1$. These binary operators, denoted by \tnorm, are used in
mathematical fuzzy logic to interpret conjunction.
%
%
The \emph{residuum}~$\implication$ is a binary operator which is used to interpret
implication. It is uniquely defined by the property that
\[ (x\tnorm y)\leqslant z \text{ iff } y\leqslant(x\implication z)
  \text{ for all } x,y,z\in\chain. \]
The \emph{residual negation}~$\negation$ is defined simply as
$\negation x:=x\implication 0$.
Finally, the \emph{t-conorm}, used for the disjunction, is 
defined as 
$x\tconorm y:=1-((1-x)\tnorm(1-y))$ for all $x,y\in\chain$.
Two prominent families of operators are based on the \emph{finite \godel 
\mbox{t-norm}}
and the \emph{finite \L ukasiewicz t-norm} (see Table~\ref{table: fuzzy logic 
operators}).
Note, however, that we do not restrict our considerations to only those logics
listed in Table~\ref{table: fuzzy logic operators}; our results are valid for
any semantics based on a finite t-norm.
\begin{table*}[tb]
  \centering
  \caption{Families of fuzzy logic operators.}
  \label{table: fuzzy logic operators}
  \begin{tabular}{>{\centering}p{0.16\textwidth}
    >{\raggedright}p{0.17\textwidth}
    >{\raggedright}p{0.17\textwidth}
    >{\raggedright}p{0.17\textwidth}
    >{\raggedright}p{0.17\textwidth}}
    \toprule
    name &
      conjunction $x\tnorm y$ &
      disjunction $x\tconorm y$ &
      negation $\negation x$ &
      implication $x\implication y$ \tabularnewline
    \midrule
    Zadeh &
      $\min(x,y)$ &
      $\max(x,y)$ &
      $1-x$ &
      $\max(1-x,y)$ \tabularnewline[1mm]
    \godel &
      $\min(x,y)$ &
      $\max(x,y)$ &
      $\begin{cases}1&\text{if }x=0\\0&\text{if }x>0\end{cases}$ &
      $\begin{cases}1&\text{if }x\leqslant y\\y&\text{if }x>y\end{cases}$
      \tabularnewline[4mm]
    \lukas &
      $\max(x+y-1,0)$ &
      $\min(x+y,1)$ &
      $1-x$ &
      $\min(1-x+y,1)$ \tabularnewline
    \bottomrule
	\end{tabular}
\end{table*}

An alternative to t-norm based approaches for interpreting the logical
connectives in fuzzy logics is the so-called Zadeh family of
operators, shown in the first row of Table~\ref{table: fuzzy logic
  operators}. Intuitively, the \emph{Zadeh family} 
can be seen as a combination of the G\"odel and the \lukas operators.

\subsection{The Fuzzy DL \chain-\SROIQ}

We introduce finitely valued, fuzzy extensions of the classical
description logic \SROIQ~\cite{horrocks2006even}---one of the most
expressive decidable DLs which provides the direct model-theoretic
semantics of the standardized ontology language for the Semantic Web OWL~2.
It has been shown that reasoning in finitely valued fuzzy extensions of \SROIQ can be 
reduced to reasoning in classical 
\SROIQ~\cite{bobillo2008optimizing,bobillo2011reasoning,straccia2004transforming}.
This reduction technique will be considered in detail in Section~\ref{sec:reduction}.

Consider three countable and pairwise disjoint sets of
\emph{individual names}~\NI, \emph{concept names}~\NC, and \emph{role
names}~\NR.
Individual names refer to single elements of an application domain, concept
names describe sets of elements, and role names
binary relations between elements.
Based on these, complex \emph{concepts} and \emph{roles} can be built using different
constructors.
More precisely, a \emph{(complex) role} is either of the form $r$ or $r^-$
(\emph{inverse role}), for $r\in\NR$, or it is the \emph{universal role}~$u$.
Similarly, \emph{(complex) concepts} are built inductively from concept names using the
following constructors
\begin{itemize}
  \item $\top$ ~ (\emph{top concept}),
  \item $\bot$ ~ (\emph{bottom concept}),
  \item $C\sqcap D$ ~ (\emph{conjunction}),
  \item $C\sqcup D$ ~ (\emph{disjunction}),
  \item $\lnot C$ ~ (\emph{negation}),
  \item $\forall r.C$ ~ (\emph{value restriction}),
  \item $\exists r.C$ ~ (\emph{existential restriction}),
  \item $\left\{d_1/a_1,\ldots,d_m/a_m\right\}$ ~ (\emph{fuzzy nominal}),
  \item \atLeast{m}{r}{C} ~ (\emph{at-least restriction}),
  \item \atMost{m}{r}{C} ~ (\emph{at-most restriction}), and
  \item $\exists r.\Self$ ~ (\emph{local reflexivity}),
\end{itemize}
where $C,D$ are concepts, $r$ is a role, $m$ is a natural number,
$a_1,\dots,a_m\in\NI$, and $d_1,\dots,d_m\in\chainp$.

An \emph{ontology} $\Onto$ consists of the intensional and the extensional knowledge
related to an application domain.
The intensional knowledge, i.e.\ the general knowledge about the application
domain, is expressed through 
\begin{itemize}\item 
  a \emph{TBox}~\Tmc, a set of finitely many \emph{(fuzzy) general concept
    inclusion (GCIs) axioms} of the form $\axiom{C\sqsubseteq D\geqslant
    d}$, where $d\in\chainp$, and

\item an \emph{RBox} \Rmc, a finite set of \emph{role axioms}, which
  are statements of the following form:
  \begin{itemize}
  \item $\axiom{r_1\dots r_m\sqsubseteq r\geqslant d}$ ~ (\emph{(fuzzy) complex
      role inclusion}),
  \item $\text{trans}(r)$ ~ (\emph{transitivity}),
  \item $\text{dis}(r_1,r_2)$ ~ (\emph{disjointness}),
  \item $\text{ref}(r)$ ~ (\emph{reflexivity}),
  \item $\text{irr}(r)$ ~ (\emph{irreflexivity}),
  \item $\text{sym}(r)$ ~ (\emph{symmetry}), or
  \item $\text{asy}(r)$ ~ (\emph{asymmetry}),
  \end{itemize}
  where $r,r_1,\dots,r_m$ are roles and again $d\in\chainp$.
\end{itemize}
The extensional knowledge, which refers to the particular knowledge about specific
facts or situations, is expressed by an \emph{ABox}~\Amc containing a finite set of statements
about individuals of the form:
\begin{itemize}
  \item \axiom{C(a)\bowtie d} ~ (\emph{concept assertion}),
  \item \axiom{r(a,b)\bowtie d} ~ (\emph{role assertion}),
  \item $a\neq b$ ~ (\emph{individual inequality assertion}), or
  \item $a=b$ ~ (\emph{individual equality assertion}),
\end{itemize}
where $a,b\in\NI$, $C$ is a concept, $r$ is a role,
$d\in\chain$, and ${\bowtie}\in\{\leqslant,\geqslant\}$.
If $\bowtie$ is $\geqslant$, then we again consider only values $d>0$; dually,
for assertions using $\leqslant$ we assume that $d<1$.

For any axiom of the form $\axiom{\alpha\geqslant 1}$, we may simply
write~$\alpha$.
Finally, an \emph{ontology} is a tuple $\Omc=(\Amc,\Tmc,\Rmc)$ consisting of an
ABox~\Amc, a TBox~\Tmc, and an RBox~\Rmc.

To ensure decidability of classical \SROIQ, a set of restrictions
regarding the use of roles is imposed. For example, transitive roles are not allowed
to occur in number restrictions (for more details see~\cite{horrocks2006even}).
The same restrictions are also adopted for fuzzy extensions of
\SROIQ~\cite{bobillo2008optimizing,bobillo2009fuzzy,bobillo2011reasoning}.
However, they are not essential for our purposes, as all results presented in
this paper hold regardless of these restrictions (except, of course, for the complexity
results of Section~\ref{sec:complexity}).

The semantics of \chain-\SROIQ is defined via interpretations. 
A \emph{(fuzzy) interpretation} is a pair $\I=(\Delta^\I,\cdot^\I)$, consisting
of a non-empty 
set $\Delta^\I$ (called the \emph{domain}) and an \emph{interpretation function} $\cdot^\I$ 
that maps every individual name~$a\in\NI$ to an element $a^\I\in\Delta^\I$, every concept 
name~$A\in\NC$ to a fuzzy set $A^\I\colon\Delta^\I\to\chain$, and every role name~$r\in\NR$ 
to a fuzzy binary relation $r^\I\colon\Delta^\I\times\Delta^\I\to\chain$.
This function is extended to complex roles and complex concepts as described in
Table~\ref{table: SROIQ concepts}.
\begin{table*}[tb]
\caption{Syntax and semantics of concepts in \chain-\SROIQ}
\label{table: SROIQ concepts}
\centering
{%
\renewcommand{\arraystretch}{1.3}
\begin{tabular}{>{\raggedright}m{0.19\textwidth}>{\raggedright}m{0.50\textwidth}}
\hline\rowcolor{Gray}
	Concept $C$ & Semantics $C^\I(x)$ for $x\in\Delta^\I$\tabularnewline
\hline
	$\top$ & $1$\tabularnewline
	$\bot$ & $0$ \tabularnewline
	$C\sqcap D$ & $C^\I(x)\tnorm D^\I(x)$ \tabularnewline
	$C\sqcup D$ & $C^\I(x)\tconorm D^\I(x)$\tabularnewline
	$\neg C$ & $\negation C^\I(x)$ \tabularnewline
	$\forall r.C$
	& $\displaystyle \inf_{y\in\Delta^\I} r^\I(x,y)\implication C^\I(y)$
	\tabularnewline
	$\exists r.C$
	& $\displaystyle \sup_{y\in\Delta^\I} r^\I(x,y)\tnorm C^\I(y)$
	\tabularnewline[3mm]
	$\left\{ d_1/o_{1},\ldots,d_m/o_{m}\right\} $
	& $\begin{cases}d_i&\text{if }x=o_i^\I,\ i\in\{1,\dots,m\}\\
	  0&\text{otherwise}\end{cases}$ \tabularnewline
	\atLeast{m}{r}{C}
	& $\displaystyle \sup_{\substack{y_1,\ldots y_m\in\Delta^\I\\
	  \text{pairwise different}}}
  	  \min_{i=1}^m r^\I(x,y_i)\tnorm C^\I(y_i)$
	\tabularnewline[4mm]
	\atMost{m}{r}{C}
	& $\displaystyle \inf_{\substack{y_1,\ldots y_{m+1}\in\Delta^\I\\
		\text{pairwise different}}}
	    \ominus\left(\min_{i=1}^{m+1} r^\I(x,y_i)\tnorm C^\I(y_i)\right)$
	\tabularnewline
	$\exists r.\Self$ & $r^\I(x,x)$ \tabularnewline
\hline\rowcolor{Gray}
	Role $s$ & Semantics $s^\I(x,y)$ for $x,y\in\Delta^\I$ \tabularnewline
\hline 
	$r^{-}$ & $r^\I(y,x)$ \tabularnewline
	$u$ & $1$ \tabularnewline
\hline
\end{tabular}
}
\end{table*}
Note that the usual fuzzy semantics of existential and value restrictions,
number restrictions, and role inclusions formally requires the computation of
an infimum or supremum over all domain elements. However, since~\chain is a
finite chain, in our case these are actually minima or maxima, respectively.

For the two-valued chain $\chain=\{0,1\}$, we obtain the semantics of classical
\SROIQ, since then all fuzzy operators correspond their classical counterparts.
In this setting, it is more natural to treat~$C^\I$ as a subset of
$\Delta^\I$, given by its characteristic function
$C^\I\colon\Delta^\I\to\{0,1\}$ (and analogously for roles). We call
such an interpretation a \emph{classical interpretation}.

An ontology is \emph{satisfied} by a fuzzy interpretation~\I if all of
its axioms are satisfied, as defined in Table~\ref{table: SROIQ
  axioms}.
\begin{table*}[tb]
\caption{Syntax and semantics of axioms in \chain-\SROIQ}
\label{table: SROIQ axioms}
\centering
{
\renewcommand{\arraystretch}{1.3}
\begin{tabular}{>{\raggedright}m{0.19\textwidth}>{\raggedright}m{0.50\textwidth}}
\hline \rowcolor{Gray}
	ABox & Semantics \tabularnewline
\hline 
	$C(a) \bowtie d$ & $C^\I(a^\I)\bowtie d$
	\tabularnewline
	$r(a,b) \bowtie d $ & $r^\I(a^\I,b^\I)\bowtie d$
	\tabularnewline
	%
	%
	$a\neq b$ & $a^\I\neq b^\I$ \tabularnewline
	$a=b $ & $a^\I = b^\I$ \tabularnewline
\hline \rowcolor{Gray}
	TBox & Semantics (for all $x\in\Delta^\I$) \tabularnewline
\hline 
	$\langle C\sqsubseteq D \geqslant d \rangle$
	& $C^\I(x)\implication D^\I(x) \geqslant d$ \tabularnewline
\hline \rowcolor{Gray}
	RBox & Semantics (for all $x,y,z\in\Delta^\I$) \tabularnewline
\hline
	$\langle r_1 \ldots  r_n\sqsubseteq r \geqslant d\rangle$
	& $\displaystyle \left(\sup_{x_1,\dots,x_{n-1}\in\Delta^\I}
	  r_1^\I(x,x_1)\tnorm\ldots \tnorm r_n^\I(x_{n-1},y)\right)
	  \implication r^\I(x,y) \geqslant d$ \tabularnewline
	$\text{trans}(r)$  & $r^\I(x,y)\tnorm r^\I(y,z)\leqslant r^\I(x,z)$
	  \tabularnewline
	$\text{dis}(r_1,r_2)$ & $r_1^\I(x,y)=0$ or $r_2^\I(x,y)=0$ \tabularnewline
	$\text{ref}(r)$ & $r^\I(x,x)=1$ \tabularnewline
	$\text{irr}(r)$ & $r^\I(x,x)=0$ \tabularnewline
	$\text{sym}(r)$ & $r^\I(x,y)=r^\I(y,x)$ \tabularnewline
	$\text{asy}(r)$ & $r^\I(x,y)=0$ or $r^\I(y,x)=0$ \tabularnewline
\hline
\end{tabular}
}
\end{table*}
In this case, \I is
called a \emph{model} of the ontology. An ontology is \emph{consistent} iff it
has a model.
In fuzzy extensions of \SROIQ, axioms are often allowed to express also strict
inequalities ($<$~and~$>$). However, in the finitely valued setting an axiom
$\axiom{\alpha>d}$ with $d<1$ can be expressed as
$\axiom{\alpha\geqslant\next{d}}$
(and similarly for~$<$ and the direct lower neighbour of~$d$).
%

In the literature it is also common to find negated role assertions of the form
$\axiom{\lnot r(a,b)\bowtie d}$~\cite{BoDG-IJUF09,BDGS-IJUF12}. However, in our
setting $\axiom{\lnot r(a,b)\leqslant d}$ is equivalent to an assertion of the form
$\axiom{r(a,b)\geqslant d'}$, and similarly for $\axiom{\lnot r(a,b)\geqslant d}$.

\begin{example}\label{ex: ontology}
  Suppose that we have a cloud computing environment consisting of
  multiple servers with their own internal memory and CPU.
	%
	To model such an environment, we use 
                \begin{itemize}
        \item 
          the individual names: \exfo{serverA}, \exfo{serverB},
          \exfo{memA}, \exfo{memB}, \exfo{cpuA}, \exfo{cpuB};
        \item   the concept names: 
          \exfo{CPU},
          \exfo{Memory}, 
          \exfo{Overused},
          \exfo{Server},
          \exfo{ServerWithLimitedResources},
          and 
          \\ \exfo{ServerWithAvailableResources}; 
          and
        \item the role names: \exfo{hasPart} and \exfo{isConnectedTo}.
        \end{itemize}%
	The assertional knowledge of this domain is modeled via the ABox~\Amc 
	\begin{align*}
		\big\{ &
		\exfo{Server}(\exfo{serverA}), \exfo{CPU}(\exfo{cpuA}),
		  \exfo{Memory}(\exfo{memA}), \\
		& \axiom{\exfo{Overused}(\exfo{cpuA})\geqslant 0.8},
		  \exfo{Overused}(\exfo{memA}), \\
		& \exfo{hasPart}(\exfo{serverA},\exfo{cpuA}),
		  \exfo{hasPart}(\exfo{serverA},\exfo{memA}), \\
		& \axiom{\exfo{ServerWithAvailableResources}(\exfo{serverB})
		  \geqslant 0.6}, \\
		& \axiom{\exfo{isConnectedTo}(\exfo{serverA},\exfo{serverB})
		  \geqslant 0.8} \big\},
	\end{align*}
	which, for example, states that \exfo{cpuA} is overused with degree at
	least~$0.8$, and that the memory \exfo{memA} is also overused with
	degree 1.
	The terminological knowledge of this domain can be modeled via a TBox~\T
	containing axioms like
	\begin{align*}
    \langle\exfo{Server}\sqcap{}
      &\exists\exfo{hasPart}.(\exfo{Overused}\sqcap\exfo{CPU})\sqcap{} \\
      &\exists\exfo{hasPart}.(\exfo{Overused}\sqcap\exfo{Memory}) \\
      &\quad\sqsubseteq\exfo{ServerWithLimitedResources}\geqslant 0.8\rangle,
  \end{align*}
	stating that a server with an overused memory and CPU 
	is a server with limited resources. This implication must hold with a degree
	of at least $0.8$.
	  
  It should be noted that the concepts \exfo{CPU}, \exfo{Memory}, and \exfo{Server} and the 
  role \exfo{hasPart} are essentially \emph{crisp}; i.e., they can only take values in $\{0,1\}$. 
  This information can be easily modeled as part of a fuzzy ontology and handled by the reduction 
  algorithm in~\cite{bobillo2008optimizing}.
  In contrast to this, the concepts \exfo{ServerWithLimitedResources}
  and \exfo{Overused} have a vague nature; that is, they are fuzzy, and the degree to
  which a server has limited resources and the degree to which a CPU
  or a memory card is overused can take values strictly between $0$
  and $1$.  The role \exfo{isConnectedTo} is also fuzzy and it is used
  to declare the connection between two servers.  The higher the connection 
  degree between two servers is, the larger bandwidth they use in
  their communication.
\end{example}

\subsection{Conjunctive Queries}

Based on the semantics other reasoning services than consistency of
ontologies can be defined. In this paper we are interested in
conjunctive query answering. We give the definition of classical
conjunctive queries next.
\begin{definition}[Conjunctive Query]
  Let \NV be a countably infinite set of \emph{variables} disjoint from~\NC,
  \NR, and \NI.  An \emph{atom} is a \emph{concept atom} of the form $A(x)$, a
  \emph{role atom} of the form $r(x,y)$, or an \emph{equality atom} of the form $x\approx y$, where
  $x,y\in\NV\cup\NI$, $A\in\NC$, and $r\in\NR$.
  A \emph{($k$-ary) conjunctive query (CQ)} $q$ is a statement of the form
  \[ (x_1,\dots,x_k)\gets\alpha_1,\dots,\alpha_m, \]
  where $\alpha_1,\dots,\alpha_m$ are atoms, and $x_1,\dots,x_k$ are (not
  necessarily distinct) variables occurring in these atoms. We call
  $x_1,\dots,x_k$ the \emph{distinguished variables} of~$q$.
  $\Var{q}$ denotes the set of all variables and individual names occurring
  in~$q$. If $k=0$, we call $q$ a \emph{Boolean conjunctive query}.

  Let $\I$ be a classical interpretation, $q$ a Boolean CQ, and
  $\pi\colon\Var{q}\rightarrow\Delta^\I$ a function such that
  $\pi(a)=a^\I$ for all $a\in\NI$.
  If $\pi(x)\in A^\I$, then we write $\I\models^\pi A(x)$, and
  $\I\models^\pi r(x,y)$ whenever $(\pi(x),\pi(y))\in r^\I$, and
  $\I\models^\pi x\approx y$ if $\pi(x)=\pi(y)$.
  If $\I\models^\pi\alpha$ for all atoms $\alpha$ in~$q$, we write
  $\I\models^\pi q$ and call $\pi$ a \emph{match} for $\I$ and~$q$.
  We say that $\I$ \emph{satisfies} $q$ and write $\I\models q$ if
  there is a match $\pi$ for $\I$ and~$q$.

  A \emph{($k$-ary) union of conjunctive queries (UCQ)} $q_\text{UCQ}$ is a set
  of $k$-ary conjunctive queries.
  An interpretation \I satisfies a Boolean UCQ $q_\text{UCQ}$, written
  $\I\models q_\text{UCQ}$ if $\I\models q$ for some $q\in q_\text{UCQ}$.
  For a Boolean (U)CQ~$q$ and an ontology~\Onto, we write $\Onto\models q$
  and say that $\Onto$ \emph{entails}~$q$ if $\I\models q$ holds for all 
  models~\I of~\Onto.

  Consider now an arbitrary $k$-ary (U)CQ~$q$ and a $k$-tuple $\ans\in\NI^k$ of
  individual names. We say that $\ans$ is an \emph{answer} to~$q$ w.r.t.\ an
  ontology~\Onto if \Onto entails the Boolean (U)CQ $\ans(q)$ resulting
  from~$q$ by replacing all distinguished variables according to~\ans (and
  possibly introducing new equality atoms if some of the distinguished
  variables in an answer tuple are equal).
\end{definition}
The problem of \emph{query answering} is to compute all answers of a (U)CQ
w.r.t.\ a given ontology.
Query answering can be reduced to query entailment by testing all possible
tuples $a\in\NI^k$, which yields an exponential blow-up.
It is well-known that query entailment and query answering can be mutually
reduced and that decidability and complexity results carry over modulo the
mentioned blow-up~\cite{calvanese1998decidability}.

\begin{example}
\label{exa:cq-instantiation}
  Consider the UCQ $\exfo{isMonitoredBy}$, consisting of the following CQs:
  \begin{align*}
    (y,x) &\gets \exfo{monitors}(x,y), \\
    (x,x) &\gets \exfo{SelfMonitored}(x).
  \end{align*}
  To obtain all answers of this UCQ, we consider all possible tuples
  $(a,b)\in\NI$ and instantiate the CQs as follows:
  \begin{align*}
    () &\gets \exfo{monitors}(b,a), \\
    () &\gets \exfo{SelfMonitored}(a),\ a\approx b,
  \end{align*}
  which results in a Boolean UCQ.
  The latter is entailed by an ontology~\Omc if one can derive that in all
  models of~\Omc either the assertion $\exfo{monitors}(b,a)$ holds, or else
  both $a\approx b$ and $\exfo{SelfMonitored}(a)$ are satisfied.
\end{example}
If the distinguished variables are clear from the context, we may also omit
them and write a CQ simply as a set of atoms.

In fuzzy DLs, conjunctive queries can be of two different types: 
\emph{threshold conjunctive queries} or \emph{general fuzzy 
queries}~\cite{pan2007expressive,straccia2013foundations,straccia2014top}.%
\footnote{In \cite{straccia2014top}, queries are defined that allow for
grouping, aggregation, and ranking. Although we do not consider such queries
here, we generalize our basic queries in Section~\ref{sec:gen:query}.}
Threshold queries ask for tuples of individuals that satisfy a set of
assertions to at least some given
degree. For example, the threshold query
\begin{multline*}
  \{\exfo{Server}(x)\geqslant 1, \exfo{hasPart}(x,y)\geqslant 1,
    \exfo{CPU}(y)\geqslant 1, \\
  \qquad \exfo{Overused}(y)\geqslant 0.6\}
\end{multline*}
asks for all pairs of servers and CPUs such that the CPU is a part
of the server and is also overused to a degree of at least~$0.6$.
\begin{definition}[Threshold Conjunctive Query]\label{def: Threshold 
Conjunctive Query}
  A \emph{degree\lb atom} is an expression of the form $\alpha\geqslant d$, 
  where
  $\alpha$ is an atom and $d\in\chainp$.
  A \emph{($k$-ary) threshold conjunctive query}~\qt is of the form
  \[ (x_1,\dots,x_k) \gets \alpha_1\geqslant d_1,\dots,\alpha_m\geqslant d_m,\]
  where $\alpha_1\geqslant d_1,\dots,\alpha_m\geqslant d_m$ are degree atoms
  and $x_1,\dots,x_k$ are variables.
  As before, $\Var{\qt}$ denotes the set of variables and individuals
  occurring in the threshold CQ $\qt$.

  Let $\I$ be an interpretation, \qt a Boolean threshold CQ, and
  $\pi\colon\Var{\qt}\rightarrow\Delta^\I$ a function that maps each
  $a\in\NI$ to $a^\I$.
  The \emph{degree} of an atom $\alpha=A(x)$ w.r.t.\ $\pi$ is defined as
  $\alpha^\Imc(\pi):=A^\Imc(\pi(x))$, and we set
  $\alpha^\Imc(\pi):=r^\Imc(\pi(x),\pi(y))$ for $\alpha=r(x,y)$; finally, for
  $\alpha=x\approx y$ we define
  $\alpha^\Imc(\pi):=1$ if $\pi(x)=\pi(y)$, and $\alpha^\I(\pi):=0$ otherwise.
  If $\alpha^\I(\pi)\geqslant d$ holds for all degree atoms $\alpha\geqslant d$
  in~\qt, then we write $\I\models^{\pi}\qt$ and call $\pi$ a \emph{match} for
  $\I$ and~$\qt$.
  The notions of satisfaction, entailment, and answers are defined as for
  classical CQs.
\end{definition}
General fuzzy CQs, in contrast, have the same syntax as classical conjunctive
queries. Their answers are the tuples of individuals satisfying them to a 
degree
greater than~$0$, together
with the degree to which the query is satisfied. For
example,
\begin{equation}
  \{\exfo{Server}(x),\exfo{hasPart}(x,y),\exfo{CPU}(y),\exfo{Overused}(y)\}
  \label{eq: server fuzzy query}
\end{equation}
asks for all overused CPUs that belong to a server, along with the degree
to which these CPUs are overused. To obtain the degree of the query from the
individual degrees of the atoms, the fuzzy operator interpreting the
conjunction is used.
\begin{definition}[Fuzzy Conjunctive Query]\label{def: Fuzzy
    Conjunctive Queries}
  A \emph{($k$-ary) fuzzy conjunctive query}~\qf is of the form
  \[ (x_1,\dots,x_k)\gets \alpha_1,\dots,\alpha_m, \]
  where $\alpha_1,\dots,\alpha_m$ are atoms and $x_1,\dots,x_k$ are variables.
	Let $\I$ be an interpretation, $\qf$ a Boolean fuzzy $\cq$, and $\pi$ a
  mapping as in Definition~\ref{def: Threshold Conjunctive Query}.
	If $\bigotimes_{\alpha\in\qf}\alpha^\Imc(\pi)\geqslant d>0$, then we write
  $\I\models^{\pi}\qf\geqslant d$ and call $\pi$ a \emph{match} for $\I$ and
  $\qf$ with a \emph{degree} of at least~$d$.
	We say that $\I$ \emph{satisfies} $\qf$ with a degree of at least~$d$ and
	write $\I\models \qf\geqslant d$ if there is such a match.
	If $\I\models \qf\geqslant d$ for all models $\I$ of an ontology~\Onto, we
	write $\Onto\models \qf\geqslant d$ and say that $\Onto$ \emph{entails}~$\qf$
	with a degree of at least~$d$.
        Finally, a tuple
        $\ans\in\NI^k$ is an \emph{answer} to a $k$-ary
        fuzzy CQ~\qf w.r.t.\ \Onto with a degree of at least~$d$ if
        \Onto entails~$\ans(\qf)$ with a degree of at least~$d$.
\end{definition}
The query entailment problem for a (Boolean) threshold $\cq$ is to decide
whether $\Onto\models\qt$. For fuzzy $\cq$s, we may consider two variants of
the query entailment problem, namely
\begin{itemize}
  \item to decide whether $\Onto\models \qf\geqslant d$ for a given
    $d\in\chainp$, or
  \item to find the best entailment degree
    $\max\{d\mid \Onto\models\qf\geqslant d\}$.
\end{itemize}
Since we consider only finitely valued semantics over the chain~\chain,
these two problems can be polynomially reduced to each other.
As for classical query answering, it suffices to analyze the complexity of
query \emph{entailment}; the results can then be transferred to query
answering~\cite{calvanese1998decidability}.

\begin{example}
  Consider the following queries:
	\begin{align*}
	  \qt &:=
	    \{\exfo{hasPart}(x,y)\geqslant 1,\exfo{Overused}(y)\geqslant 0.9\},\\
	  \qf &:=
	    \{\exfo{hasPart}(x,y),\exfo{Overused}(y)\},
	\end{align*}
	and the ontology from Example~\ref{ex: ontology}.
	An answer to the query~\qt is $(\exfo{serverA},\exfo{memA})$, but not
	$(\exfo{serverA},\exfo{cpuA})$ since $\exfo{cpuA}$ is only overused to
	degree~$0.8$.
	The answers to~\qf are the pairs $(\exfo{serverA},\exfo{cpuA})$ with
	degree $\geqslant 0.8$ and $(\exfo{serverA},\exfo{memA})$ to degree~$1$.
\end{example}
%

\begin{remark}
  A threshold CQ with inequalities using~$\leqslant$ would correspond to a
  classical CQ containing negated role atoms, for which query answering is
  undecidable even in very inexpressive DLs~\cite{GIKK-RR13,Rosa-ICDT07}.
  Similarly, upper bounds for fuzzy conjunctive queries~\qf, i.e.\ asking
  whether $\Omc\models\qf\leqslant d$, can be seen as a generalized form of
  disjunction of (negated) query atoms.
  For these reasons, we consider only inequalities using $\geqslant$.
\end{remark}

Before we turn to answering such queries over fuzzy
ontologies, we describe the reduction of expressive finitely valued fuzzy
ontologies to classical ones.

%% file: 3.reduction.tex

\section{Reduction of Finitely Valued Fuzzy Ontologies to Classical
  Ontologies}
\label{sec:reduction}

A popular reasoning technique for fuzzy DLs based on finite chains is
the reduction of the fuzzy ontology to a classical one. This allows to
use existing DL systems to reason in the fuzzy description logic.
However, a major drawback of existing approaches for finite chains
using arbitrary t-norms (see~\cite{BDGS-IJUF12,bobillo2011reasoning,MaPe-RR-14}) is
that this reduction introduces an exponential blow-up in the size of
the fuzzy ontology. While this handicap can be remedied by our normalization
step described in Section~\ref{sec:normalization} (see also the experiments in
Section~\ref{sec:eval-norm}), another obstacle needs to be addressed first: the
reduction proposed in~\cite{BDGS-IJUF12,bobillo2011reasoning,MaPe-RR-14} is not
correct for number restrictions.
In the following, we describe this problem in detail and propose a (partial)
solution.

\subsection{Treating Number Restrictions}
\label{subsec:number-restrictions}

The reduction in~\cite{BDGS-IJUF12,bobillo2011reasoning,MaPe-RR-14} is based 
on the idea to simulate number
restrictions by existential restrictions in the following way.
For a number restriction $\atLeast{m}{r}{C}$, the new concept names
$B_1,\dots,B_m$ and the axioms
\begin{itemize}
\item $\top\sqsubseteq B_1\sqcup\dots\sqcup B_m \text{ and}$
\item $B_i\sqcap B_j\sqsubseteq\bot \text{ for all }i,j,\ 1\leqslant
  i<j\leqslant m.$
\end{itemize}
are introduced, which require them to form a partition.
Subsequently, the number restriction $\atLeast{m}{r}{C}$ is replaced by
the concept
$\exists r.(C\sqcap B_1)\sqcap\ldots\sqcap\exists r.(C\sqcap B_m)$.
The following classical example shows that this replacement does not
preserve the semantics of the number restrictions, and thus cannot
be correct in the fuzzy case, either.
\begin{example}
\label{exa:number-restrictions-incorrect}
	Consider the following ABox and TBox:
	\begin{align*}
	  \A :=
	    \{& r(a,a),\,r(a,b),\,r(b,a),\,r(b,c),\,r(c,b),\,r(c,c), \\
	      & a\neq b,\,b\neq c,\,a\neq c \}, \\
	  \T :=
	    \{& \top\sqsubseteq\atMost{2}{r}{\top},\,
	      \top\sqsubseteq\atLeast{2}{r}{\top} \}.
	\end{align*}
  A simple model for $\Omc=(\A,\T, \emptyset)$ is given by $\Delta^\I:=\{a,b,c\}$ and
  $r^\I:=\{(a,a),(a,b),(b,a),(b,c),(c,b),(c,c)\}$. Thus,
  the ontology $\Omc$ is consistent.
	By replacing $\atLeast{2}{r}{\top}$ according to the method described above,
	we obtain the TBox
	\begin{align*}
	  \T' :=
	    \{& \top\sqsubseteq\atMost{2}{r}{\top},\,
	      \top\sqsubseteq\exists r.B_1\sqcap\exists r.B_2, \\
	    & B_1\sqcap B_2\sqsubseteq\bot,\,
	      \top\sqsubseteq B_1\sqcup B_2 \}.
	\end{align*}
	We show that the resulting ontology $(\A,\T',\emptyset)$ is
        inconsistent. Assume to the contrary that there exists a model $\I'$ of
        the ontology $\Omc'= (\A, \T', \emptyset)$. Without loss of
        generality, suppose that it interprets the individual names $a,b,c$ as
        themselves. Thus, we must have $r^\I\subseteq r^{\I'}$.
	There are only eight possible combinations for $a,b,c$ belonging to either
	$B_1$ or $B_2$.
	Suppose first that $a,b\in B_1^{\I'}$ and $c\in B_2^{\I'}$. Then by the axiom
	$\top\sqsubseteq\exists r.B_1\sqcap\exists r.B_2$ the individual~$a$ must
	have yet another $r$-successor $x\in B_2^{\I'}$.
	However, this contradicts the GCI $\top\sqsubseteq\atMost{2}{r}{\top}$.
	Similar arguments apply for all other combinations, and therefore the
	ontology is inconsistent.
\end{example}
It should be noted that for \godel and Zadeh semantics alternative (correct)
reductions of number restrictions exist~\cite{bobillo2009fuzzy,BDGS-IJUF12}.

We now propose an alternative encoding of number restrictions when using other
fuzzy semantics, avoiding the problem exhibited by
Example~\ref{exa:number-restrictions-incorrect}.
Intuitively, instead of using a partition of the target concept~$C$ of a
restriction $\atLeast{m}{r}{C}$, we will partition the role~$r$.
Note first that at-most restrictions can be expressed using negation and
at-least restrictions; that is, $\atMost{m}{r}{C}$ has the same semantics as
$\lnot(\atLeast{(m+1)}{r}{C})$ (cf.\ Table~\ref{table: SROIQ concepts}).
Hence, in the following we focus on methods for handling at-least number
restrictions.
Furthermore, we can assume without loss of generality that they only occur in
axioms of the forms
\[ \axiom{A\sqsubseteq\atLeast{m}{r}{B}\geqslant d} \quad \text{and} \quad
  \axiom{\atLeast{m}{r}{B}\sqsubseteq A\geqslant d}, \]
where $A$ and $B$ are concept names (cf.\ Section~\ref{sec:normalization}).

Axioms of the first kind can be equivalently expressed using $m$ fresh role
names $r_1,\dots,r_m$ in the following axioms:
\[ \axiom{A\sqsubseteq\exists r_i.C \geqslant d}, \quad
  r_i\sqsubseteq r, \quad
  \text{dis}(r_j,r_k), \]
for all $i,j,k\in\{1,\dots,m\}$ with $j<k$. This is correct due to the minimum
used in the semantics of at-least restrictions.
More precisely, every model of the original axiom can be extended by a
suitable interpretation of the new role names to a model of the resulting
axioms, and every model of the latter is immediately a model of the former.
Hence, we can eliminate all at-least restrictions that occur on the right-hand
side of GCIs (and all at-most restrictions that occur on the left-hand side of
GCIs).

Unfortunately, this approach does not work for at-least restrictions occurring
on the left-hand side of GCIs. The reason is that the presence of $m$ many
$r$-successors satisfying~$C$ does not imply that these successors can be
reached using one of the disjoint roles $r_1,\dots,r_m$. However, this can be
expressed using the additional role axiom
\begin{equation}
\label{eq:role-disjunction}
  r\sqsubseteq r_1\sqcup\dots\sqcup r_m,
\end{equation}
which involves a \emph{role disjunction} that is interpreted using the maximum,
i.e.,
\[ (r_1\sqcup\dots\sqcup r_m)^\Imc(x,y):=\max_{i=1}^m r_i^\Imc(x,y). \]
Role disjunction is an example of a (safe) Boolean role constructor, which can
be added to most classical DLs without increasing the complexity of
reasoning~\cite{RuKH-JELIA08}.
Moreover, some query answering procedures for classical DLs even work in the
presence of such constructors~\cite{CaEO-IJCAI09}.
Unfortunately, to the best of our knowledge, role disjunctions are not yet
supported by any  classical DL reasoner.

In the presence of axiom~\eqref{eq:role-disjunction} and the role disjointness
axioms from above, the GCI
$\axiom{\atLeast{m}{r}{C}\sqsubseteq A \geqslant d}$ can now be equivalently
expressed as
\[ \axiom{\exists r_1.C\sqcap\dots\sqcap\exists r_m.C \sqsubseteq A
  \geqslant d}. \]
Unlike the incorrect reduction for number restrictions that was first
proposed in~\cite{bobillo2011reasoning}, our approach does not partition the
\emph{range} of the role~$r$ in the number restriction, but rather the role 
itself, and hence it
correctly treats the case where a domain element is an $r$-successor of two
different elements that are subject to the same number restriction on~$r$
(recall Example~\ref{exa:number-restrictions-incorrect}).
However, like the approach
of~\cite{BDGS-IJUF12,bobillo2011reasoning,MaPe-RR-14}, this incurs an
exponential blow-up in the largest number occurring in number restrictions, if
these numbers are represented in the ontology using a binary encoding. 
The reduction is polynomial if we assume unary encoding of numbers.

Since role disjunctions are not supported by \SROIQ or OWL\,2, we will
restrict the following investigation to \emph{unqualified} number restrictions
of the form $\uatLeast{m}{r}:=\atLeast{m}{r}{\top}$ and
$\uatMost{m}{r}:=\atMost{m}{r}{\top}$, i.e., to the fuzzy logic \chain-\SROIN.
However, we want to emphasize that we can easily treat qualified number
restrictions in the following reduction if the classical target language
supports role disjunctions.
It is straightforward to extend the reduction to deal even with \SROIQBs, the
extension of \SROIQ with full Boolean role expressions (which satisfy a safety
condition)~\cite{RuKH-JELIA08}; the reduction of the role constructors is
similar to the one for concepts.

\subsection{Ontology normalization for \chain-\SROIN ontologies}
\label{sec:normalization}

The reason that the reductions described
in~\cite{bobillo2009fuzzy,BDGS-IJUF12,bobillo2011reasoning} can cause an
exponential blow-up in the size of the ontology is that concept constructors
may be nested to arbitrary depths. In this subsection, we propose a
normalization step to ensure that each GCI and concept assertion contains at
most one concept constructor, and that each complex role inclusion contains at
most two roles on the left-hand side.
Because of this, the subsequent reduction of a \chain-\SROIN ontology~\Onto to
a classical \SROIN ontology~$\Onto\crisp$ causes only a linear blow-up in the
size of~\Onto (and a quadratic blow-up in the size of~\chain).
For an experimental evaluation of the resulting difference in ontology size
and reasoning performance, see Section~\ref{sec:eval-norm}.

The normalization proceeds by exhaustively replacing each axiom by a
set of axioms according to Table~\ref{table: ontology normalization}.
\begin{table*}[tb]
\centering
\caption{Normalization rules for \chain-\SROIN ontologies}
\label{table: ontology normalization}
{
\renewcommand{\arraystretch}{1.2}
\begin{tabular}{%
  >{\raggedleft}p{0.17\textwidth}%
  @{\ \ $\leadsto$\ \ }%
  >{\raggedright}p{0.28\textwidth}%
  >{\raggedleft}p{0.14\textwidth}%
  @{\ \ $\leadsto$\ \ }%
  >{\raggedright}p{0.28\textwidth}}
\toprule
$\axiom{C\sqsubseteq D\geqslant d}$ &
$\axiom{C\sqsubseteq A_D\geqslant d}$,~\ $A_D\sqsubseteq D$ &

$\axiom{C(a)\geqslant d}$ &
$\axiom{A_C(a)\geqslant d}$,~$A_C\sqsubseteq C$ \tabularnewline

$\axiom{A\sqsubseteq C\sqcap D\geqslant d}$ &
$\axiom{A\sqsubseteq A_C\sqcap A_D\geqslant d}$,~$A_C\sqsubseteq C$,~ %
  $A_D\sqsubseteq D$ &

$\axiom{C\sqcap D\sqsubseteq A\geqslant d}$ &
$\axiom{A_C\sqcap A_D\sqsubseteq A\geqslant d}$,~$C\sqsubseteq A_C$,~%
  $D\sqsubseteq A_D$ \tabularnewline

$\axiom{A\sqsubseteq C\sqcup D\geqslant d}$ &
$\axiom{A\sqsubseteq A_C\sqcup A_D\geqslant d}$,~$A_C\sqsubseteq C$,~%
  $A_D\sqsubseteq D$ &

$\axiom{C\sqcup D\sqsubseteq A\geqslant d}$ &
$\axiom{A_C\sqcup A_D\sqsubseteq A\geqslant d}$,~$C\sqsubseteq A_C$,~%
  $D\sqsubseteq A_D$ \tabularnewline

$\axiom{A\sqsubseteq\neg C\geqslant d}$ &
$\axiom{A\sqsubseteq\neg A_C\geqslant d}$,~$C\sqsubseteq A_C$ &

$\axiom{\neg C\sqsubseteq A\geqslant d}$ &
$\axiom{\neg A_C\sqsubseteq A\geqslant d}$,~$A_C\sqsubseteq C$ \tabularnewline

$\axiom{A\sqsubseteq\exists r.C\geqslant d}$ &
$\axiom{A\sqsubseteq\exists r.A_C\geqslant d}$,~$A_C\sqsubseteq C$ &

$\axiom{\exists r.C\sqsubseteq A\geqslant d}$ &
$\axiom{\exists r.A_C\sqsubseteq A\geqslant d}$,~$C\sqsubseteq A_C$
  \tabularnewline

$\axiom{A\sqsubseteq\forall r.C\geqslant d}$ &
$\axiom{A\sqsubseteq\forall r.A_C\geqslant d}$,~$A_C\sqsubseteq C$ &

$\axiom{\forall r.C\sqsubseteq A\geqslant d}$ &
$\axiom{\forall r.A_C\sqsubseteq A\geqslant d}$,~$C\sqsubseteq A_C$
  \tabularnewline

$\axiom{r_1r_2r_3\ldots r_m\sqsubseteq r\geqslant d}$ &
$\axiom{r_{r_1r_2}r_3\ldots r_m\sqsubseteq r\geqslant d}$,~ %
  $r_1r_2\sqsubseteq r_{r_1r_2}$ \tabularnewline
\bottomrule
\end{tabular}
}
\end{table*}
In that table, $A,A_C,A_D$ denote concept names, $\top$, or~$\bot$; $C,D$ are
complex concepts that are neither concept names, $\top$, nor~$\bot$; and
$r_1,\dots,r_m$, $r$, and~$r_{r_1r_2}$ are roles.
$A_C$ and $A_D$ are fresh concept names that abbreviate the concepts~$C$
and~$D$, respectively.
In the last rule, $r_{r_1r_2}$ is a fresh role name that stands for the role
composition of~$r_1$ and~$r_2$.
For simplicity, we have given the rules for conjunctions and
disjunctions only for the case where both operands are complex
concepts. However, if only one of them is a complex concept, we would
not introduce a new concept name for the other operand.
Note that nominals, unqualified number restrictions, and local reflexivity
concepts do not need to be normalized.

It should be noted
that this reduction is not correct under Zadeh semantics due to the
properties of the implication function.
However, \cite{BoDG-IJUF09} provides a different reduction for this case that
does not exhibit an exponential blow-up even without normalization.
Hence, we consider in the following result only semantics that are based on
finitely valued t-norms and their induced operators \implication, \negation,
and \tconorm.
\begin{proposition}
\label{prop:normalization-correct}
  Let $\Onto'$ be the ontology resulting from the exhaustive application of the
  rules in Table~\ref{table: ontology normalization} to a \chain-\SROIN
  ontology~\Onto.
  Under t-norm based semantics, every model of~\Onto can be extended to a
  model of~$\Onto'$ by interpreting the new concept names~$A_C$ like~$C$ and
  $r_{r_1r_2}$ like~$r_1r_2$.%
\footnote{Where $(r_1r_2)^\I(x,z):=
  \sup_{y\in\Delta^\I}r_1^\I(x,y)\tnorm r_2^\I(y,z)$
  (cf.\ Table~\ref{table: SROIQ axioms}).}
  Moreover, every model of~$\Onto'$ is already a model of~\Onto.
\end{proposition}
This simple observation immediately shows that $\Onto'$ is consistent iff \Onto
is consistent. Moreover, it allows us to prove correctness of the normalization
procedure also with respect to the other reasoning tasks we will consider in
the following sections.
Furthermore, it is easy to see that the normalization could be extended to deal
also with qualified number restrictions (\Qmc).

While this procedure involves the introduction of linearly many new concept
names, it allows us to circumvent the exponential blow-up exhibited by previous
reductions.

\begin{remark}
  The reason why this normalization reduces the complexity of the following
  reduction is that it ensures that each axiom contains at most three
  occurrences of concept or role names.
  However, we will see in the following subsection that concept and role names
  that are interpreted classically, i.e.\ can take only the values~$0$ and~$1$,
  do not take part in the reduction.
  Hence, it is enough to ensure that each axiom contains at most three
  occurrences of \emph{fuzzy} concept or role names. Such axioms do not need to
  be reduced any further.
  Nevertheless, all complexity results concerning the reduction in the
  following section remain valid.
\end{remark}

\begin{example}
  The normalized form of the TBox \T containing the GCI from
  Example~\ref{ex: ontology} is as follows:
  \begin{align*}
    \{ &\exfo{Overused}\sqcap\exfo{CPU}\sqsubseteq A, \\
    &\exfo{Overused}\sqcap\exfo{Memory}\sqsubseteq B, \\
    &\exists\exfo{hasPart}.A \sqsubseteq C, \displaybreak[0]\\
    &\exists\exfo{hasPart}.B \sqsubseteq D, \displaybreak[0]\\
    &\exfo{Server}\sqcap C\sqsubseteq E, \\
    &\langle E\sqcap D \sqsubseteq \exfo{ServerWithLimitedResources}
      \geqslant 0.8\rangle \}
  \end{align*}
  However, since the GCI in \T contains only three occurrences of
  names of fuzzy concepts (two times \exfo{Overused} and once
  \exfo{ServerWithLimitedResources}), we can use \T as it is in the
  following reduction.
\end{example}
We will assume in the following that \Onto is already normalized.
The remainder of the reduction is very similar to the one described
in~\cite{BDGS-IJUF12} (except for number restrictions).

\subsection{The Reduction Algorithm}
\label{subsec:reduction}

Each concept name and role name in~\Onto is mapped onto a set of concepts and roles
corresponding to their \emph{$\alpha$-cuts}, which are crisp sets containing
all elements that belong to a fuzzy set to at least a given degree $\alpha$.
For example, if the concept name $\exfo{Overused}$ describes the degree to
which a CPU is overused, then $\exfo{Overused}_{\geqslant 0.6}$
represents the set of CPUs that are overused to a degree of at least~$0.6$.
It is clear that we do not need to consider the value~$0$ for such cuts, as
$A_{\geqslant 0}$ always describes the whole domain.
We may also refer to concept names of the form $A_{>d}$ for $d\in\chain$ and
$d<1$, which is a short-hand notation for $A_{\ge\next{d}}$, and similarly for
role names.

The $\SROIN$ ontology $\Onto\crisp$ obtained from the reduction has the
following form:
\begin{itemize}
  \item To preserve the semantics of $\alpha$-cuts of concept and role names,
    the following axioms are added to~$\Onto\crisp$ for all $A\in\NC$,
    $r\in\NR$, and $d\in\chain$ with $0<d<1$:
    \[
      A_{>d}\sqsubseteq A_{\geqslant d},\ r_{>d}\sqsubseteq r_{\geqslant d}.
    \]
  \item Each complex concept~$C$ appearing in~\Onto is mapped to the
    complex concept $\rho(C,\geqslant d)$ that represents its
    $\alpha$-cut regarding degree $d$, as defined in the first part of
    Table~\ref{table:reduction} in the appendix.
  \item Each axiom in~\Onto is then mapped to a classical axiom or set of
    axioms in $\Onto\crisp$ according to the mapping~$\kappa$ defined in the
    second part of Table~\ref{table:reduction}.
\end{itemize}
For a more detailed analysis of the reduction rules, the interested reader may refer
to~\cite{bobillo2009fuzzy,BDGS-IJUF12,bobillo2011reasoning}.
We provide a detailed proof of correctness in the appendix.
\begin{theorem}
\label{thm:crispification-correct}
  Let \Onto be a \chain-\SROIN ontology. Then \Onto has a fuzzy model
  iff its reduced form $\Onto\crisp$ has a classical model.
\end{theorem}
Our normalization procedure allows us to show the following improved complexity
bounds. The proof of the following lemma can be found in Appendix~\ref{sec:ap-lem13}.
\begin{lemma}
\label{lem:crispification-complexity}
  For a normalized \chain-\SROIN ontology~\Onto, the size of $\Onto\crisp$ is
  linear in the size of~\Onto and quadratic in the size of~\chain.
\end{lemma}
This means that, by simply introducing the normalization step, we can avoid
the exponential blow-up of the crispification approach. In particular,
we greatly improve the exponential bounds shown
in~\cite{BDGS-IJUF12,bobillo2011reasoning}.

\begin{example}
\label{ex:red}
  Figure~\ref{fig:ex-9} contains the reduced form of the ontology from
  Example~\ref{ex: ontology} w.r.t.\ \lukas semantics over the chain with six
  elements $\mathscr{C}=\{0,0.2,0.4,0.6,0.8,1\}$. We have taken into account
  that one does not need to consider $\alpha$-cuts of classical concept and
  role names.
  This nicely illustrates how classical concepts and roles help to reduce the
  size of the reduction. Not only do we have $1$ crisp concept instead of
  $n-1$ cut concepts, but the number of disjunctions and conjunctions
  introduced can be reduced dramatically (cf.\ 
  Table~\ref{table: SROIQ reduction}).
\end{example}

\begin{figure*}[tb]
  \centering
  \begin{boxedminipage}[t]{\textwidth}
  \vspace*{-\baselineskip}
  \begin{align*}
    \A\crisp := \big\{ & \exfo{Server}(\exfo{serverA}),~ \exfo{CPU}(\exfo{cpuA}),~
    \exfo{Memory}(\exfo{memA}),~ 
    \exfo{Overused}_{\geqslant 0.8}(\exfo{cpuA}),~
    \exfo{Overused}_{\geqslant 1}(\exfo{memA}),~ \exfo{hasPart}(\exfo{serverA},\exfo{cpuA}), \\
    & \exfo{hasPart}(\exfo{serverA},~ \exfo{memA}),~ \nonumber
    \exfo{ServerWithAvailableResources}_{\geqslant 0.6}(\exfo{serverB}),~ 
    \exfo{isConnectedTo}_{\geqslant
      0.8}(\exfo{serverA},~ \exfo{serverB}) \big \}
    \\[\bigskipamount]
    \T\crisp := \{ & \exfo{Overused}_{\geqslant 0.4}\sqsubseteq
    \exfo{Overused}_{\geqslant 0.2}, \dots \} ~ \cup{} \\
    &\bigcup_{\substack{d_1,d_2\text{ minimal, }d_3\text{ maximal}\\
        \text{such that }(d_1\tnorm d_2)\implication d_3<0.8}}
    \hspace{-5ex}
    \big\{ \exfo{Server}\sqcap
    \exists\exfo{hasPart}.(\exfo{Overused_{\geqslant d_1}}
    \sqcap \exfo{CPU})\sqcap{} 
    \exists\exfo{hasPart}.(\exfo{Overused_{\geqslant d_2}}
    \sqcap \exfo{Memory}) \sqsubseteq \exfo{ServerWithLimitedResources_{>d_3}}
    \big\}
  \end{align*}
  \end{boxedminipage}
  \caption{ABox and TBox for Example~\ref{ex:red}}\label{fig:ex-9}
\end{figure*}

%% file: 5.cqs_answering.tex

\section{Conjunctive Query Answering for Fuzzy DLs}
\label{sec: reduction based conjunctive query answering}

In this section we show how to solve the problem of answering threshold and
fuzzy CQs in finitely valued fuzzy DLs by taking advantage of existing
algorithms for answering unions of conjunctive queries in classical DLs.
Our solution is based on the reduction technique described in
Section~\ref{sec:reduction} .

\subsection{Translating Fuzzy and Threshold CQs}

In the following, we define a function $\kappa$ that maps each threshold CQ and
fuzzy CQ to a (U)CQ in a classical DL.
The idea is that we can then evaluate these classical queries
over~$\Onto\crisp$ in order to answer the original queries over~\Onto.
The shape of the mapping~$\kappa$ depends on the type of query, and uses
$\alpha$-cuts for reducing fuzzy concept and role names to classical ones
(cf.\ Section~\ref{sec:reduction}).

We first define the function $\kappa$ for degree atoms of queries analogously as it was
done for assertions:
\begin{align*}
  \kappa(A(x)\geqslant d) &:= A_{\geqslant d}(x), \\
  \kappa(r(x,y)\geqslant d) &:= r_{\geqslant d}(x,y).
\end{align*}
This definition is then lifted to threshold CQs~\qt in the obvious way:
$\kappa(\qt):=\{\kappa(\alpha) \mid \alpha\in\qt\}$.

For transforming fuzzy {\cq}s, we use $\alpha$-cuts as in the translation of 
concepts. Thus, in the 
case of fuzzy {\cq}s, $\kappa$ receives as input also a membership degree from \chainp, which will
be a lower bound for the degree of the query. Recall
that the answers  depend on the operator \tnorm that interprets the
conjunction in the logic
under consideration. Since there can be several different combinations
of degrees for the different atoms that result in the same degree, fuzzy {\cq}s are translated into
unions of (classical) conjunctive queries, representing each of these
combinations.

Formally, for a fuzzy CQ $\qf=\{\alpha_1,\dots,\alpha_n\}$ and $d\in\chainp$,
$\kappa(\qf,\geqslant d)$ is the set of all conjunctive queries
$$\{\kappa(\alpha_1\geqslant d_1),\dots,\kappa(\alpha_n\geqslant d_n)\},$$ where
$d_1,\ldots,d_n\in\chain$ are such that $\bigotimes_{i=1}^n d_i\geqslant d$.


\begin{example}
	Consider a threshold CQ asking for all pairs of connected servers such that the first one has limited 
	and the second one has available resources:
	\begin{align*}
		&\{\exfo{ServerWithLimitedResources}(x)\geqslant 0.8, \\
		&\qquad \exfo{isConnectedTo}(x,y)\geqslant 0.6, \\
		&\qquad \exfo{ServerWithAvailableResources}(y)\geqslant 0.6\}.
	\end{align*}
  This threshold \cq{} is reduced to the following classical \cq:
	\begin{align*}
		&\{\exfo{ServerWithLimitedResources}_{\geqslant 0.8}(x), \nonumber\\
		&\qquad \exfo{isConnectedTo}_{\geqslant 0.6}(x,y), \nonumber \\
		&\qquad \exfo{ServerWithAvailableResources}_{\geqslant 0.6}(y)\}.
	\end{align*}
	A fuzzy \cq asking for the same information, but without the thresholds, is the following:
	\begin{align*}
		&\{\exfo{ServerWithLimitedResources}(x), \\
		&\qquad\exfo{isConnectedTo}(x,y), \\
		&\qquad\exfo{ServerWithAvailableResources}(y)\}.
	\end{align*}
	To acquire all the pairs $(x,y)$ that satisfy this query with degree at
	least~$0.8$, for the \lukas t\mbox{-}norm over the chain with 6 membership
	degrees, the query is reduced to the following union of classical \cq{s}:
	\begin{align*}
		\big\{
		&\{
		\exfo{ServerWithLimitedResources}_{\geqslant 0.8}(x), \\
		&\qquad \exfo{isConnectedTo}_{\geqslant 1}(x,y), \\
		&\qquad \exfo{ServerWithAvailableResources}_{\geqslant 1}(y)
		\}, \nonumber\\
		&\{
		\exfo{ServerWithLimitedResources}_{\geqslant 1}(x), \\
		&\qquad \exfo{isConnectedTo}_{\geqslant 0.8}(x,y), \\
		&\qquad \exfo{ServerWithAvailableResources}_{\geqslant 1}(y)
		\}, 
		\nonumber\\
		&\{
		\exfo{ServerWithLimitedResources}_{\geqslant 1}(x), \\
		&\qquad \exfo{isConnectedTo}_{\geqslant 1}(x,y), \\
		&\qquad \exfo{ServerWithAvailableResources}_{\geqslant 0.8}(y)
		\}\big\}.
	\end{align*}
\end{example}
The following theorem states that our query reduction is sound and complete.

\begin{theorem}\label{theorem: correctness of query reduction}
  Let $\Onto\crisp$ be the classical version of the fuzzy ontology~$\Onto$,
  $\qt$ be a threshold $\cq$, $\qf$ be a fuzzy $\cq$, and $d\in\chain$.
  Then the following equivalences hold:
  \begin{enumerate}
  \item  $\Onto \models \qt~\Leftrightarrow~\Onto\crisp	\models \kappa(\qt)$
  \item $\Onto \models \qf\geqslant d~\Leftrightarrow~\Onto\crisp
    \models \kappa(\qf,\geqslant d)$.
  \end{enumerate}
\end{theorem}
\begin{proof}
  In order to prove that $ \Onto\crisp \models \kappa(\qt)$ implies
  $\Onto\models \qt$, consider any fuzzy model $\I$ of $\Onto$. By
  Proposition~\ref{prop:normalization-correct}, \Imc can be extended to a model
  of the normalized ontology $\Omc'$. We now define the classical interpretation
  $\J=\{\Delta^\J,\cdot^\J\}$ as follows (cf.\ Appendix~\ref{app:completeness}):
  \begin{align}
    \Delta^{\J} & := \Delta^{\I} \nonumber \\
    a^\J & := a^\I \nonumber \\
    A_{\geqslant d}^{\J} & :=
      \left\{ \beta \mid A^\I(\beta )\geqslant d \right\} \nonumber \\
    r_{\geqslant d}^{\J} & :=
      \left\{ (\beta,\gamma) \mid r^\I(\beta ,\gamma)\geqslant d \right\}.
    \label{eq:I->J main document}
  \end{align}
      By Lemma~\ref{lem: completeness}, $\J$ is a classical model
      of~$\Onto\crisp$.
      Since $\J\models\Onto\crisp$ and $ \Onto\crisp \models \kappa(\qt)$, it
      follows that $\J\models \kappa(\qt)$.
      By the construction of~$\J$, it can be easily
      verified that $\I\models \qt$.

      It can be shown in a similar way that $\Onto \models \qf\geqslant d$
      implies $\Onto\crisp \models\kappa(\qf,\geqslant d)$.
	
  To prove the opposite direction we build for
  each classical model \J of $\Onto\crisp$, the fuzzy interpretation 
  $\I=(\Delta^\I,\cdot^\I)$, where
  \begin{align}
   \Delta^{\I} & := \Delta^{\J} \nonumber \\
   a^\I & := a^\J \nonumber \\
   A^\I(\beta) & :=
     \max\big\{ d \mid \beta\in A_{\geqslant d}^\J\big\} \nonumber \\
   r^{\I}(\beta,\gamma) & :=
     \max\big\{ d \mid (\beta,\gamma)\in r_{\geqslant d}^\J \big\}.
   \label{eq:J->I main document}
  \end{align}
  (cf.\ Appendix~\ref{app:soundness}). By
  Proposition~\ref{prop:normalization-correct} and Lemma~\ref{lem: soundness},
  \I is a fuzzy model of the original ontology~\Onto.
	It is straightforward to show that $\Jmc\models\kappa(\qt)$ whenever
	$\Imc\models\qt$, and similarly for fuzzy CQs.
	\qed
\end{proof}
It is easy to see that this result applies not only to the reduction described
in Section~\ref{sec:reduction}, but to any reduction that can be shown correct
using the definitions in~\eqref{eq:I->J main document}
and~\eqref{eq:J->I main document}, e.g.\ the one optimized for \chain-\SROIQ
under \godel and Zadeh semantics in~\cite{BDGS-IJUF12}.

\subsection{Complexity Results}
\label{sec:complexity}

By Lemma~\ref{lem:crispification-complexity}, the size of the crispified ontology
$\Onto\crisp$ is polynomial in the size of~\Onto. Therefore, we can transfer
all complexity results for answering classical CQs over classical sublogics of
\SROIN directly to the query answering problem for threshold CQs.
In particular, recall that CQ entailment in \SROIN can be decided in
\ThreeExpTime~\cite{CaEO-IJCAI09}, is \TwoExpTime-complete for \SHIN,
and \ExpTime-complete for \ALCHN~\cite{ELOS-IJCAI09,Lutz-IJCAR08}.

Moreover, for fuzzy DLs where correct crispification algorithms exist even for
\SROIQ, e.g.\ for \godel and Zadeh semantics~\cite{BDGS-IJUF12} or in the
presence of role disjunctions (see Section~\ref{subsec:number-restrictions}),
any classical CQ answering technique that is able to handle number restrictions
can be applied also for threshold CQs.
Under this condition, answering threshold CQs in finitely valued extensions of
\SHIQ and \SHOQ is \TwoExpTime-complete~\cite{ELOS-IJCAI09,GlHS-KR08}, for
\SROQ and \SRIQ it can be done in \ThreeExpTime~\cite{CaEO-IJCAI09}, while for
\SHQ it becomes \ExpTime-complete if queries are restricted to
\emph{simple} roles, i.e.\ roles that do not have transitive
subroles~\cite{Lutz-IJCAR08}.
Notice, however, that none of these query answering approaches has been implemented
so far.

For fuzzy CQs, the complexity increases by an exponential factor due to the
blow-up in the translation $\kappa$. It is possible to eliminate this blow-up,
however, when the minimum t-norm is used, i.e.\ under \godel and Zadeh
semantics. Then we can obviously define
$\kappa(\qf,\geqslant d):=\{\kappa(\alpha\geqslant d) \mid \alpha\in\qf\}$
for any fuzzy CQ~\qf, and thus obtain the same complexity results as for
threshold CQs.

It should also be noted that the \emph{data complexity} of all these problems
is the same as that for classical CQs, as the size of the ABox is not increased
by the reduction from~\Onto to~$\Onto\crisp$. Since for many applications the TBox remains unchanged, while the ABox changes frequently, the reduction to the crisp ontology $\Onto\crisp$ need not be computed when queries are answered, but need be computed only once ``off-line'' beforehand.

\subsection{Generalizing the Query Component}
\label{sec:gen:query}


So far, we have examined the reduction technique for answering threshold CQs 
and fuzzy CQs.
These two types of queries are immediate extensions of classical CQs, and
usually considered in the literature. 
Nevertheless, the existence of degrees may lead to more general forms of fuzzy 
CQs 
in which the score of a query is computed via a monotone scoring function, as 
described in 
the following.

\begin{definition}
  A \emph{($k$-ary) scoring query}~\qs is an expression of the form
  \[ (x_1,\dots,x_k) \gets f(\alpha_1,\dots,\alpha_n), \]
  where $f$ is a monotonically increasing \emph{scoring function} with $n$
  arguments, $\alpha_1,\dots,\alpha_n$ are atoms, and $x_1,\dots,x_k$ are
  variables.
Let $\I$ be an interpretation, $\qs$ a Boolean scoring query, and $\pi$ a
mapping as 
in Definition~\ref{def: Threshold Conjunctive Query}.
If
\[ f(\alpha_1^\I(\pi),\dots,\alpha_n^\I(\pi))\geqslant d, \]
then we write
$\I\models^{\pi}\qs\geqslant d$ and call $\pi$ a \emph{match} for $\I$ and
$\qs$ with a \emph{score} of at least~$d$.
We say that $\I$ \emph{satisfies} $\qs$ with a score of at least~$d$ and
write $\I\models \qs\geqslant d$ if there is such a match.
If $\I\models \qs\geqslant d$ for all models $\I$ of an ontology~\Onto, we
write $\Onto\models \qs\geqslant d$ and say that $\Onto$ \emph{entails}~$\qs$
with a score of at least~$d$.
Finally, a tuple $\ans\in\NI^k$ is an \emph{answer} to a $k$-ary
scoring query~\qs w.r.t.\ \Onto with a score of at least~$d$ if \Onto
entails~$\ans(\qs)$ with a score of at least~$d$.%
\footnote{Possibly new equality atoms $a\approx b$ introduced by the
instantiation can be connected with a multiplication to the score of the
original query (see Example~\ref{exa:cq-instantiation}).}
\end{definition}
It should be noted that the score may take an arbitrary  value in $\mathbb{R}$.
This kind of queries has already been considered in the
literature~\cite{pan2007expressive,straccia2013foundations,straccia2014top}.
Fuzzy CQs can be seen as special scoring queries of the form
$\alpha_1\tnorm\dots\tnorm\alpha_n$.
Since \chain is finite, the same technique as for fuzzy CQs can be applied
here, i.e.\ considering all possible combinations of degrees in $\chain$.

	\begin{example}
	\label{exa:agg:query}
		Suppose that we are interested in finding all servers that have overused
		CPU and memory,
		but the excessive use of CPU should be considered of greater importance 
		than the use of memory
		memory. To achieve this, we formulate the following query to include a
		weighting factor on the degrees of overuse for the different components.
		For instance, we can use the query				
		\begin{multline}\label{eq: aggregate query}
		 \exfo{Server}(x) \cdot \exfo{hasPart}(x,y) \cdot \exfo{CPU}(y) \cdot
		 \exfo{hasPart}(x,z) \cdot{} \\
		\qquad \exfo{Memory}(z) \cdot
		  \frac{3\cdot\exfo{Overused}(y)+2\cdot\exfo{Overused}(z)}{5}
		\end{multline}
	where the fraction in the last factor takes into account the degrees of
	overuse of CPU and memory with weights $0.6$ and $0.4$, respectively.

		Assume that $\chain=\{0,0.25,0.5,0.75,1\}$ and that the concepts
		$\exfo{Server}$, $\exfo{CPU}$, and $\exfo{Memory}$ and  the role
		$\exfo{hasPart}$ behave classically.
		If we want to find all answers that satisfy this query to degree
		at least $0.25$, then we can translate it into a
	    union of classical conjunctive queries that contains,~e.g.\ 
	\begin{alignat*}{1}
		  \{\exfo{Server}_{\geqslant 1}(x), &\exfo{hasPart}_{\geqslant 1}(x,y), \exfo{CPU}_{\geqslant 1}(y), \\
		  & \exfo{hasPart}_{\geqslant 1}(x,z),  \exfo{Memory}_{\geqslant 1}(z), \\
		  & \exfo{Overused}_{\geqslant 0.25}(y), \exfo{Overused}_{\geqslant 0.75}(z) \}.
	\end{alignat*}
	When evaluated over the reduced ontology, this query returns all triples of
	elements from 
	$\exfo{Server},\exfo{CPU}$, and $\exfo{Memory}$, respectively, where the CPU
	is overused to degree
	$0.25$ and the memory is overused to degree $0.75$. These tuples satisfy the
	query~\eqref{eq: aggregate query} 
	with a degree of at least
	\[ 0.6\cdot 0.25+0.4\cdot 0.75=0.45\geqslant 0.25, \]
	as desired.
	\end{example}
	Another interesting problem, specific to weighted logics, is the top-$k$
	query answering problem presented
	in~\cite{straccia2006towards,straccia2013foundations,straccia2014top}.
	This  variation of the fuzzy query answering problem focuses on the $k$ answers with the highest degrees of satisfaction.
	 In a naive approach to solve this problem, the translation function~$\kappa$
	 for fuzzy or scoring \cq{}s can be iteratively applied starting from the 
	 highest to
	 the lowest degrees in $\chain$ until the limit of $k$ answers is reached.
	 It has to be investigated if a more sophisticated approach can be adopted to solve this problem.

%% file: 5.implementation.tex

\section{Practical Evaluation}
\label{sec:results}

	We have proved that, by introducing the normalization step, we can avoid the exponential blow-up of the 
	earlier crispification approach (Section~\ref{sec:reduction}). 
	We have additionally extended the reduction process to handle the problem of CQ answering 
	(Section~\ref{sec: reduction based conjunctive query answering}).
	The main objectives of this section are to
	\begin{enumerate}
		\item evaluate how the ontology normalization preprocessing improves the total execution time for CQ 
			answering, and
	    \item study the practical limitations of the reduction approach for the problem of CQ answering in 
			the finitely valued fuzzy setting.
	\end{enumerate}

	Rather than implementing a full fuzzy CQ answering system from scratch, we have modified the 
	existing reasoner \delorean%
\footnote{\url{http://webdiis.unizar.es/~fbobillo/delorean}}\cite{BoDG-ESA12}.
	%
	\delorean is a reduction-based fuzzy DL reasoner that supports fuzzy 
	variants of the description logics 
	$\mathcal{SROIQ}(\mathcal{D})$ and 
	$\mathcal{SHOIN}(\mathcal{D})$ under the finite \godel, finite \lukas, and Zadeh semantics.
	We use \delorean to transform an input fuzzy ontology into a crisp ontology through 
	the transformation
	rules described earlier in this paper.
  Additionally, we have implemented the reduction for threshold queries described in
  Section~\ref{sec: reduction based conjunctive query answering}.
	After the crisp ontology and query are obtained, an arbitrary crisp 
	reasoner can be employed for query answering. For our experiments
  we used the \pagoda system\footnote{\pagoda version from 23rd of April 2015
  commit id: \url{30b5afef93} in
  \url{https://github.com/yujiaoz/PAGOdA}}~\cite{ZCGNH-DL-15,ZNCH-AAAI14}.
%

	Other reasoners that support query answering for expressive DLs could also
	have been adopted.
	However, most other systems can only correctly answer CQs that can be
	directly expressed in the DL, e.g., tree-shaped CQs or instance queries, or
	employ the simplifying restriction that the non-distinguished variables can
	only be bound to named individuals \cite{MNOW-KI-13}.
	Although in theory \pagoda also depends on the query answering capabilities
	of \hermit~\cite{GHM+-JAR14}, it was shown in~\cite{ZCGNH-DL-15} that it
	correctly answers the standard queries of the LUBM benchmark (which we use in
	the following), without having to rely on \hermit.
  This is also the case for the normalized, fuzzified, and crispified versions
  of the LUBM TBox we use in the following experiments.

\subsection{Test Data and Test Set-up}
	
	In order to evaluate our approach, we used the LUBM%
\footnote{\url{http://swat.cse.lehigh.edu/projects/lubm/}}
  ontology benchmark~\cite{guo2005lubm}.
	The LUBM benchmark contains  terminological knowledge that describes the 
	domain of a university, and generates synthetic OWL data (ABox) over the 
	ontology specifying individuals belonging to the university.
	The ABox corresponding to a single university contains approximately $1300$ 
	concept assertions and 
	$2450$ role assertions and can be scaled by a factor of $k$, producing 
	information for $k$ different 
	universities.
  In our experiments, we used different ABoxes covering~$1$, $15$, and $30$
  universities, respectively.

	To use LUBM as a benchmark for fuzzy reasoning, we have extended all ABox
	axioms with random degrees chosen from a fixed finite chain of cardinality
	$3$, $7$, or~$11$, which are interpreted using \lukas semantics. The axioms
	in the TBox are always required to hold with degree~$1$.
	Note that the original LUBM TBox is formulated in $\mathcal{ELHI}_{R^+}$ (we
	ignored all datatype axioms), and hence the same holds for our fuzzified
	variants. However, due to the introduction of disjunctions in the reduction
	to a classical ontology, the resulting TBox is formulated in $\mathcal{SHI}$
	(see Table~\ref{table: SROIQ reduction}).

  The queries of the LUBM benchmark were translated to threshold queries by
  asking each atom to hold with the smallest non-zero degree ($\frac{1}{2}$,
  $\frac{1}{6}$, and $\frac{1}{10}$ for $3$, $7$, and $11$ truth degrees,
  respectively).

  For all experiments, we report the average runtime over $5$ separate runs;
  the coefficient of variation was mostly below 25\%, which means that the runtimes did not
  differ much between runs.
%
  All the tests were performed on a machine powered by an Intel Core i7 2.6 GHz
  processor and equipped with 8 GB DDR3 1600MHz main memory.

\subsection{Evaluating the Normalization and its Effect on Query Answering}
\label{sec:eval-norm}

	For the first experiment, we applied the normalization from
	Section~\ref{sec:normalization} to the LUBM TBox, and compared the query
	answering times for a single challenging threshold CQ (based on query \#9 of
	the LUBM benchmark).
	%
	The results of this evaluation, based on the \lukas semantics, are presented 
	in Table~\ref{tab:norm:eval}.
\def\tableheading#1{\rotatebox{75}{\begin{minipage}[t]{0.12\textwidth}#1\end{minipage}}}%
\begin{table}
\caption{Time needed for preprocessing by \pagoda (in s) and query execution
time (in ms) for instance queries over the original LUBM TBox vs.\ its
normalized version}
\label{tab:norm:eval}
\centering
\resizebox{\columnwidth}{!}{
\begin{tabular}{rrr||rr|rr|rr|}
 &&& \multicolumn{6}{c|}{$\sharp$ Universities}
	\\
	&&& 
	\multicolumn{2}{c|}{1} & 
	\multicolumn{2}{c|}{15} & 
	\multicolumn{2}{c|}{30} 
	\\ \cline{4-9} 
	\multicolumn{1}{@{\hspace{2mm}}r@{\hspace{-2mm}}}{\tableheading{$\sharp$ 
		Degrees}} & 
  \multicolumn{1}{@{\hspace{2mm}}r@{\hspace{-2mm}}}{\tableheading{Normalized}}&
  \multicolumn{1}{@{\hspace{2mm}}r@{\hspace{-1mm}}||}{\tableheading{TBox size}}
	  & \tableheading{Preprocessing\\ \mbox{}~~(in s)}\hspace*{-3mm}
	  & \tableheading{Querying (in ms)}\hspace*{-1mm}
	  & \tableheading{Preprocessing\\ \mbox{}~~(in s)}\hspace*{-3mm}
	  & \tableheading{Querying (in ms)}\hspace*{-1mm}
	  & \tableheading{Preprocessing\\ \mbox{}~~(in s)}\hspace*{-3mm}
	  & \tableheading{Querying (in ms)}\hspace*{-1mm} \\
	\hline\hline
	3 & no & 493 & 8.94 & 2.8 & 9.24 & 9.8 & 9.52 & 18.8 \\
	\rowcolor{lightgray}
	3 & yes & 518 & 8.92 & 2.4 & 9.21 & 9.4 & 9.52 & 16.6 \\
	\hline
	7 & no & 2187 & 13.09 & 2.4 & 13.58 & 5.0 & 14.25 & 7.2 \\
	\rowcolor{lightgray}
	7 & yes & 2022 & 12.21 & 2.2 & 12.57 & 3.8 & 13.06 & 6.4 \\
	\hline
	11 & no & 5177 & 26.62 & 2.0 & 27.36 & 3.2 & 28.45 & 5.0 \\
	\rowcolor{lightgray}
	11 & yes & 3942 & 18.69 & 2.2 & 19.60 & 3.6 & 20.50 & 5.6
 \tabularnewline
\hline
\end{tabular}
}
\end{table}
The table compares the sizes of the crispified TBoxes for both the original
TBox and the normalized one (counted as the number of occurrences of concept
and role names), as well as the running times of \pagoda.
It shows both the time needed by \pagoda for its own preprocessing of the
ontology and the time needed to perform query answering.
The preprocessing clearly dominates the overall runtime. In comparison to this,
the time needed by \pagoda to actually answer the queries is almost negligible.

The preprocessing time is mostly determined by the normalization and by the
number of degrees, i.e., the size of \chain, but it also increases slightly
with increased ABox size.
One can also see that, although the query answering times are not much
affected, the size of the TBox, and hence the time \pagoda needs for
preprocessing, can be reduced significantly by normalization.
This can also be seen in the graph in Figure~\ref{fig:tbox-sizes}, where the
normalized variant of the TBox is clearly  smaller  w.r.t.\ 
increasing numbers of degrees than the original TBox.

%
However, for only~$3$ degrees of membership, the avoided exponential blow-up does not
outweigh the overhead of introducing auxiliary concept names for the
normalization in practice.
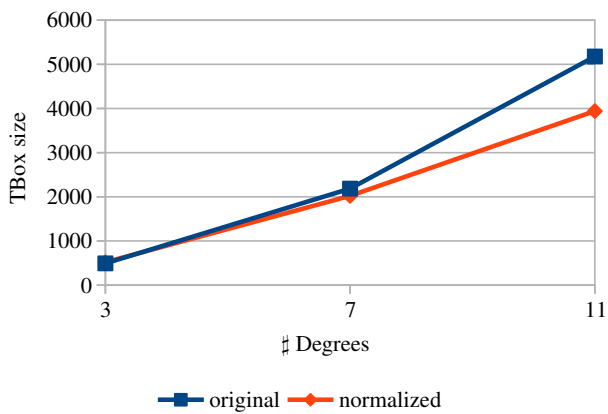
\begin{figure}[tb]
  \centering
  \def\svgwidth{\columnwidth}
  \input{TBox_size.txt}
  \caption{The sizes of the crispified TBoxes (the number of occurrences of
    concept and role names) compared between the original LUBM TBox and its
    normalized variant}
  \label{fig:tbox-sizes}
\end{figure}

Finally, it should be noted that the query answering time actually \emph{decreases}
when the number of truth degrees is increased (for both the normalized and the
original TBox). This is not relevant for our experiments, but is simply a
consequence of the \lukas semantics:
When more fine-grained truth degrees are randomly assigned to all ABox
assertions, the chance is larger that the degree to which a query atom is
satisfied is reduced to~$0$, and hence does not yield an answer tuple (see
Table~\ref{table: fuzzy logic operators}). And with a decreased number of
answers, the queries need less time to be answered.

\subsection{Evaluating Conjunctive Query Answering }
	In order to analyze the practical usefulness of the reduction approach for CQ
	answering, we consider several CQs from the LUBM benchmark (namely \#5, \#7,
	and \#9) over the normalized LUBM TBox.
	An additional parameter that is evaluated here is the percentage of crisp
	concepts and roles appearing in the ontology; 0\%, 20\%, 80\%, or 100\% of
	all names are randomly chosen to be crisp.
	Recall that a crisp concept or role name can only take the values $0$ or $1$.
	For such concepts and roles, the reduction process is more light-weight, 
	since it results in a smaller number of concept names in the reduced 
	ontology~\cite{BoDG-ESA12}. 
	For each chain size, number of universities, and percentage of crisp symbols,
	we evaluated three \emph{threshold CQs}%
\footnote{Note that for answering fuzzy CQs instead of threshold CQs, the
runtime will increase by an exponential factor, depending on the number of
atoms in the fuzzy CQ.}
	that extend the CQs in the LUBM benchmark.
	The evaluation was performed again w.r.t. the \lukas semantics. The
	results are depicted in Table~\ref{tab: query Lukas}.
\definecolor{lightgray}{rgb}{0.89, 0.89, 0.89}
\definecolor{darkgray}{rgb}{0.75,0.75,0.75}
\def\queries{4}
\begin{table*}
\caption{Time needed for preprocessing by \pagoda (in s) and query
execution time (in ms) for three queries and different percentages of crisp
names}
\label{tab: query Lukas}
\centering
\begin{tabular}{rrr||r@{\hspace{2.2mm}}rrr|r@{\hspace{2.2mm}}rrr|r@{\hspace{2.2mm}}rrr|}
 &&& \multicolumn{12}{c|}{$\sharp$ Universities}
	\\
	&&&
	\multicolumn{\queries}{c|}{1} & 
	\multicolumn{\queries}{c|}{15} & 
	\multicolumn{\queries}{c|}{30}
	\\  \cline{4-15}
	\multicolumn{1}{@{\hspace{2mm}}r@{\hspace{-2mm}}}{\tableheading{$\sharp$ 
	Degrees}}&
	\multicolumn{1}{@{\hspace{2mm}}r@{\hspace{-2mm}}}{\tableheading{Crisp 
	$\%$}}&
	\multicolumn{1}{@{\hspace{2mm}}r@{\hspace{-1mm}}||}{\tableheading{TBox 
	size}}&
	\tableheading{\parbox{2cm}{Preprocessing\\ \mbox{}~~(in s)}}\hspace*{-4mm}&
	\hspace*{1mm}\tableheading{Query 5 (in ms)}\hspace*{-4mm}&
	\hspace*{1mm}\tableheading{Query 7}\hspace*{-4mm}&
	\hspace*{1mm}\tableheading{Query 9}\hspace*{-4mm}&
	\tableheading{\parbox{2cm}{Preprocessing\\ \mbox{}~~(in s)}}\hspace*{-4mm}&
	\hspace*{1mm}\tableheading{Query 5 (in ms)}\hspace*{-4mm}&
	\hspace*{1mm}\tableheading{Query 7}\hspace*{-4mm}&
	\hspace*{1mm}\tableheading{Query 9}\hspace*{-4mm}&
	\tableheading{\parbox{2cm}{Preprocessing\\ \mbox{}~~(in s)}}\hspace*{-4mm}&
	\hspace*{1mm}\tableheading{Query 5 (in ms)}\hspace*{-4mm}&
	\hspace*{1mm}\tableheading{Query 7}\hspace*{-4mm}&
	\hspace*{1mm}\tableheading{Query 9}\hspace*{-4mm}
	\\ \hline\hline
- & 100 & 207 &
8.36&     7.2&     10.2&     3.2&     
8.52&     23.6&     112.4&     24.8&     
8.64&     32.6&     172.8&     44.4
\\\hline
3 & 0 & 518 &
8.94&     4.8&     6.8&     3.0&     
9.36&     18.2&     73.6&     9.6&     
9.76&     34.6&    102.2&     18.2
\\
\rowcolor{lightgray}
3 & 20 & 412 &
8.77&     4.6&     8.0&     2.2&     
9.13&      18.6&     87.2&     9.2&     
9.32&      29.6&     139.6&     22.2
\\
\rowcolor{darkgray}
3 & 80 & 267 &
9.10&     8.8&     11.2&     3.2&     
8.78&     26.6&     130.2&     23.2&     
9.16&     37.4&     173.6&     39.4
\\\hline
7 & 0 & 2022 &
11.88 &     5.4&     3.0&     2.4&     
12.49&    17.2&     14.2&     3.6&     
13.05&     21.8&     29.0&     6.4
\\\rowcolor{lightgray}
7 & 20 & 1412 &
10.36&     3.8&     5.8&     2.4&     
11.03&     14.6&     50.2&     4.2&     
11.56&     17.4&     91.2&     7.6
\\\rowcolor{darkgray}
7 & 80 & 519 &
8.78&     8.0&     10.8&     2.8&     
9.12&     23.8&     132.0&     26.8&     
9.49&     36.8&     162.6&     42.4
\\\hline
11 & 0 & 3942 &
18.63&     2.2&     3.4&     2.6&     
19.52&     10.6&     11.4&     3.6&     
20.17&     13.2&     17.6&     4.4
\\\rowcolor{lightgray}
11 & 20 & 2700 &
13.38&     2.8&     4.6&    2.8&     
14.20&     13.4&     40.6&     3.4&     
14.71&     14.8&     60.2&     5.0
\\\rowcolor{darkgray}
11 & 80 & 803 &
9.22&     8.0&     10.6&     3.0&     
9.47&     27.0&     113.4&     23.2&     
9.70&     40.6&     169.2&     40.8
\\
	%
\hline
\end{tabular}
\end{table*}

We again see a counter-intuitive behavior, namely that the query answering time
decreases with smaller percentage of crisp elements (and hence increased TBox
size). The reason lies again with the \lukas semantics: queries over ``more
fuzzy'' TBoxes will have less answers, and hence less effort is required to
answer them.
Nevertheless, we can see that the times required to answer threshold CQs are of
the same order of magnitude as for classical CQs. Again the total runtime is
dominated by the preprocessing time of \pagoda, and hence mainly by the TBox
size. This is good news, since the preprocessing step only needs to be performed once as long as the input TBox does not change.

Moreover, our approach is scalable for large ABoxes (as much as this is the
case for a crisp ontology) since the crispification mainly affects the size of
the TBox, but there is nearly no difference in size between a fuzzy ABox and
its crispified version (see also Table~\ref{table: SROIQ reduction}).

As a final note, in practical applications of fuzzy ontologies one would
hardly expect
the ontology to have only fuzzy concept and role names, but rather that a
smaller number of fuzzy names complements a large number of crisp names that
describe precise knowledge. This is closer to our 80\% scenario, which does not
differ much from the purely crisp case in terms of runtime.

%% file: TBox_size.txt
\begingroup%
  \makeatletter%
  \providecommand\color[2][]{%
    \errmessage{(Inkscape) Color is used for the text in Inkscape, but the package 'color.sty' is not loaded}%
    \renewcommand\color[2][]{}%
  }%
  \providecommand\transparent[1]{%
    \errmessage{(Inkscape) Transparency is used (non-zero) for the text in Inkscape, but the package 'transparent.sty' is not loaded}%
    \renewcommand\transparent[1]{}%
  }%
  \providecommand\rotatebox[2]{#2}%
  \ifx\svgwidth\undefined%
    \setlength{\unitlength}{314.64567871bp}%
    \ifx\svgscale\undefined%
      \relax%
    \else%
      \setlength{\unitlength}{\unitlength * \real{\svgscale}}%
    \fi%
  \else%
    \setlength{\unitlength}{\svgwidth}%
  \fi%
  \global\let\svgwidth\undefined%
  \global\let\svgscale\undefined%
  \makeatother%
  \begin{picture}(1,0.67567562)%
    \put(0,0){\includegraphics[width=\unitlength]{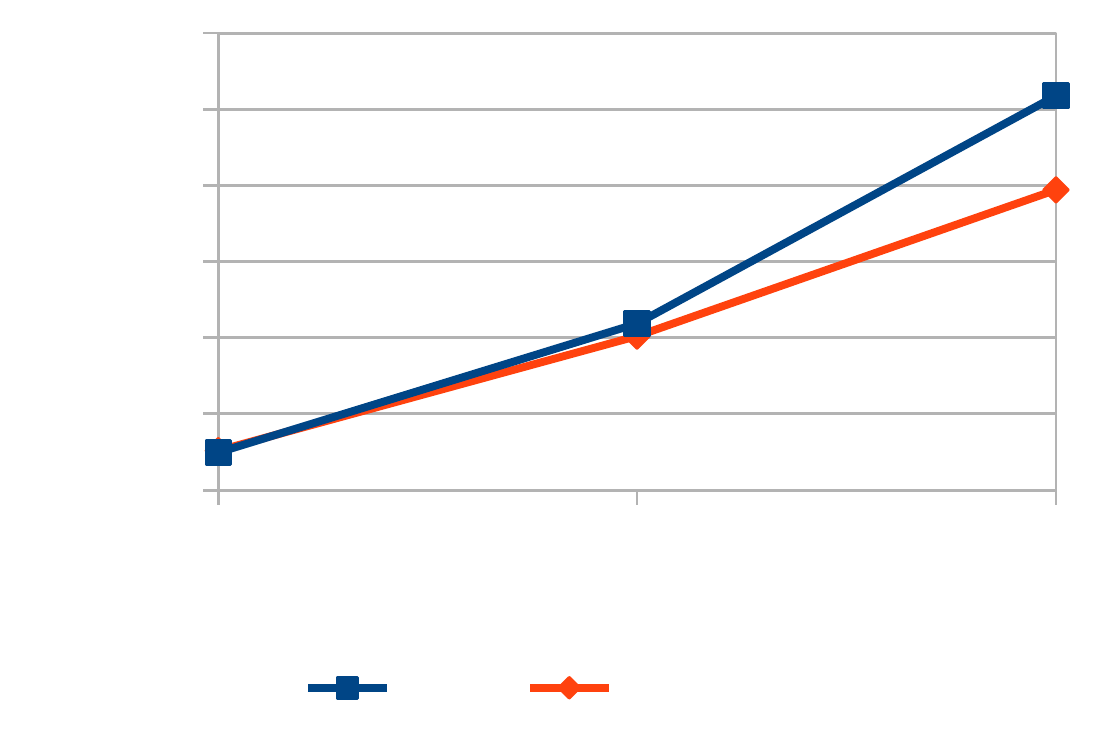}}%
    \put(0.19171174,0.17612611){\makebox(0,0)[lb]{\smash{3}}}%
    \put(0.57513513,0.17612611){\makebox(0,0)[lb]{\smash{7}}}%
    \put(0.95072068,0.17612611){\makebox(0,0)[lb]{\smash{11}}}%
    \put(0.16126129,0.21540539){\makebox(0,0)[lb]{\smash{0}}}%
    \put(0.10801805,0.2852252){\makebox(0,0)[lb]{\smash{1000}}}%
    \put(0.10801805,0.35495492){\makebox(0,0)[lb]{\smash{2000}}}%
    \put(0.10801805,0.42468465){\makebox(0,0)[lb]{\smash{3000}}}%
    \put(0.10801805,0.49441437){\makebox(0,0)[lb]{\smash{4000}}}%
    \put(0.10801805,0.5641441){\makebox(0,0)[lb]{\smash{5000}}}%
    \put(0.10801805,0.63396391){\makebox(0,0)[lb]{\smash{6000}}}%
    \put(0.36297298,0.03468468){\makebox(0,0)[lb]{\smash{original}}}%
    \put(0.56612612,0.03468468){\makebox(0,0)[lb]{\smash{normalized}}}%
    \put(0.47522523,0.12369368){\makebox(0,0)[lb]{\smash{$\sharp$ Degrees}}}%
    \put(0.07405409,0.35450447){\rotatebox{90}{\makebox(0,0)[lb]{\smash{TBox size}}}}%
  \end{picture}%
\endgroup%

%% file: 6.related_work.tex

\section{Related Work}\label{sec: related work}

	Non-fuzzy representations of fuzzy DLs have been extensively studied for
	several families of languages.
	The methods can be classified based on their fuzzy and DL expressivity used.
	Based on the Zadeh family of fuzzy logic operators, reduction techniques and  optimizations have been 
	examined for the fuzzy extensions of the $\mathcal{ALCH}$~\cite{straccia2004transforming}, 
	$\mathcal{SHOIN}$~\cite{StSt-OWLED07}, and $\mathcal{SROIQ}$~\cite{bobillo2008optimizing} languages, 
	while an experimental evaluation of the reduction technique for the fuzzy version of  $\mathcal{SHIN}$ 
	is presented in~\cite{cimiano2008reasoning}.  
	For \godel semantics, a reduction procedure for the DL \SROIQ  is
	considered in~\cite{bobillo2009fuzzy}. This procedure is extended in~\cite{BDGS-IJUF12} for a 
	language combining  \godel and Zadeh semantics.
	The reduction technique for the \lukas t-norm over \SROIQ is studied
	in~\cite{bobillo2011reasoning}, and for arbitrary finite t-norms
	in~\cite{BDGS-IJUF12,bobillo2013finite}. It should be noted, as we show in
	Example~\ref{exa:number-restrictions-incorrect}, that the
	reduction of qualified number restrictions proposed 
	in~\cite{BDGS-IJUF12,bobillo2011reasoning,MaPe-RR-14} is incorrect.
This family of algorithms has been implemented in the DeLorean
reasoner~\cite{BoDG-13,BoDG-ESA12,BoDS-TFS13}.
	Based on a different approach, a family of fuzzy DLs using $\alpha$-cuts as atomic concepts and 
	roles is considered in~\cite{li2005family}.

Conjunctive query answering for fuzzy DLs has mostly been studied for
the \textit{DL-Lite} family of languages.
In~\cite{straccia2006answering,straccia2006towards} the problem of
evaluating \mbox{top-$k$} queries in fuzzy \textit{DL-Lite} is considered.
In~\cite{pan2007expressive}, the authors present a variety of query languages
that can be used for querying fuzzy \textit{DL-Lite} ontologies and adapt
classical query rewriting techniques for answering these queries.
A similar approach is taken in~\cite{MaTu-JIST-14}.

A tableaux algorithm for conjunctive query answering for 
fuzzy CARIN, a language combining the DL
$\mathcal{ALCNR}$ with Horn rules under Zadeh semantics, is provided
in~\cite{mailis2010expressive}.
Another algorithm for answering expressive fuzzy conjunctive queries is
presented in~\cite{cheng2009deciding,cheng2009alcn}.
The algorithm allows the occurrence of both lower bound and the upper bound of
thresholds in a query atom over the DL \mbox{$\mathcal{SHIN}$} extended with
the Zadeh semantics.
Finally, practical approaches for storing and querying fuzzy knowledge
in the Semantic Web have been also investigated~\cite{simou2008storing}.

%% file: 7.conclusions.tex

\section{Conclusions and Future Work}\label{sec: conlusions}

	This paper focuses on how a classical representation of ontologies written in finitely valued
	fuzzy DLs can be adopted in order to solve the threshold and fuzzy conjunctive query answering problems.
	These problems are reduced to equivalent (U)CQ answering problems in
	classical DLs.
	The correctness of the suggested technique is proved and its complexity is
	studied for different variants of \chain-\SROIQ.
	As far as we know, no similar theoretical results have been presented.
	The proofs rely on the fact that each model of a fuzzy ontology $\Onto$ can be mapped to a model of 
	its reduced crisp form $\Onto\crisp$ and vice versa, thus showing the
	soundness and 
	completeness of the reduction technique.

	To verify the correctness of our approach, we have corrected and extended the
	proofs sketched in~\cite{bobillo2011reasoning}, providing the first full
	proof for this result in the literature.
	Additionally, we have improved  the size of the reduction from fuzzy \SROIN 
	ontologies to classical 
	\SROIN ontologies by introducing a linear normalization step. After this 
	step, the obtained classical ontology
	can be exponentially smaller than the one obtained by previous crispification approaches. The smaller size promises lower runtimes of the query answering procedure.
	%
%
	We have also suggested a reduction of qualified number restrictions, which
	uses role disjunctions to express a partitioning of a role.

  In our evaluation of a prototype implementation of this approach, based on
  the DeLorean system, we have demonstrated that the overhead involved in using
  a finitely valued fuzzzy description logic instead of a classical DL is
  reasonable.
  In a scenario where only the data and queries change frequently, while the
  TBox remains fixed, the normalization and the reduction of a finitely valued
  TBox can be computed ``off-line''. It is thus sufficient to consider the time
  required for query answering, which does not increase significantly when
  using more than two membership degrees.

Future work involves examining if available optimizations techniques for fuzzy 
and classical DLs can be applied to improve the performance of these 
algorithms.
On the theoretical side, we will investigate the possibility of extending the 
reduction to deal with qualified number restrictions without resorting to
non-standard role constructors, obtaining a polynomial reduction for any
finitely valued fuzzy description logic extending OWL~2~DL.

%% file: 8.appendix.tex


\appendix

\section{Proof of Theorem~\ref{thm:crispification-correct}}
\label{app:reduction-correct}

{
\renewcommand{\arraystretch}{1.2}

\begin{table*}	
\centering
\caption{Mapping of concepts, roles, and axioms to classical \SROIN}
\label{table: SROIQ reduction}\label{table:reduction}
\begin{tabular}{rl}
\toprule
	$\rho(\top,\geqslant d)$	& $\top$ \tabularnewline 
	$\rho(\bot,\geqslant d)$	& $\bot$ \tabularnewline
	$\rho(A,\geqslant d)$	&		$A_{\geqslant d}$		\tabularnewline
	%
	$\rho(\neg C,\geqslant d)$		&	
	$\neg\rho(C,> \negation^-(d))$
	\tabularnewline[1mm]
	%
	$\rho(C \sqcap D,\geqslant d)$						&
	$\bigsqcup_{(d_1,d_2)\in{\tnorm^-}(d)}\big(\rho(C,\geqslant d_1)\sqcap 
	\rho(D,\geqslant d_2)\big)$ 	 
	\tabularnewline[3mm]
	$\rho(C \sqcup D,\geqslant d)$						&
	$\bigsqcup_{(d_1,d_2)\in{\tconorm^-}(d)}\big(\rho(C,\geqslant d_1)\sqcap 
	\rho(D,\geqslant d_2)\big)$ 
	\tabularnewline[3mm]
	%
	$\rho(\exists r.C,\geqslant d)$						&
	$\bigsqcup_{(d_1,d_2)\in{\tnorm^-}(d)}
	  \exists \rho(r,\geqslant d_1).\rho(C,\geqslant d_2)$
	\tabularnewline[3mm]
	$\rho(\forall r.C,\geqslant d)$ &
	$\bigsqcap_{(d_1,d_2)\in{\implication^-}(d)}
	  \forall\rho(r,\geqslant d_1).\rho(C,> d_2)$
	\tabularnewline[2mm]
	%
	$\rho(\{d_1/o_1,\ldots, d_m/o_m\},\geqslant d)$ 	& 
	$\left\{o_i\mid d_i\geqslant d , i\in\{1,\dots,m\}\right\}$
\tabularnewline
  $\rho(\uatLeast{m}{r},\geqslant d)$ &
  $\uatLeast{m}{\rho(r,\geqslant d)}$
\tabularnewline
  $\rho(\uatMost{m}{r},\geqslant d)$ &
  $\uatMost{m}{\rho(r,>\negation^-(d))}$
\tabularnewline
	%
	%
	$\rho(\exists r. \text{Self},\geqslant d)$					&	
	$\exists \rho(r,\geqslant d).\text{Self}$
\tabularnewline
	$\rho(r,\geqslant d)$		& 
	$r_{\geqslant d}$
\tabularnewline
	%
	%
	$\rho(r^{-},\geqslant d)$ 		&
	$r_{\geqslant d}^{-}$
	\tabularnewline
	%
	$\rho(u,\geqslant d)$ 			& 
	$u$
\tabularnewline
\midrule

	$\kappa(C( a ) \geqslant d)$		&
	$\rho(C,\geqslant d)(a)$\tabularnewline
	$\kappa(C( a ) \leqslant d)$		&
	$\lnot\rho(C,> d)(a)$\tabularnewline
	$\kappa(r(a,b) \geqslant d)$	&
	$\rho(r,\geqslant d)(a,b)$\tabularnewline
	$\kappa(r(a,b) \leqslant d)$	&
	$\lnot\rho(r,> d)(a,b)$\tabularnewline
	$\kappa(a\neq b)$				&
	$a\neq b$						
	\tabularnewline
	$\kappa(a = b)$				&
	$a = b$						
	\tabularnewline
	$\kappa(\langle C\sqsubseteq D \geqslant d\rangle)$	&	
	$\displaystyle \bigcup_{(d_1,d_2)\in{\implication^-}(d)}
	  \big\{\rho(C,\geqslant d_1)\sqsubseteq \rho(D,>d_2)\big\}$
	\tabularnewline
	$\kappa(\langle r_1 r_2 \sqsubseteq r \geqslant d\rangle)$	&
	
	$\displaystyle \bigcup_{(d_1,d')\in{\implication^-}(d),\ 
	  (d_2,d_3)\in\implication^-(\next{d'})}
	  \big\{\rho(r_1,\geqslant d_1) \rho(r_2,\geqslant d_2)
	    \sqsubseteq \rho(r,> d_3)\big\}$
	\tabularnewline
	$\kappa(\axiom{r_1\sqsubseteq r_2\geqslant d})$ &
	$\kappa(\axiom{r_1u\sqsubseteq r_2\geqslant d})$
	\tabularnewline
	$\kappa(\text{trans}(r))$ &
	$\kappa(rr\sqsubseteq r)$
	\tabularnewline
	$\kappa(\text{dis}(r_1,r_2))$				&
	$\text{dis}(\rho(r_1,>0),\rho(r_2,>0))$
	\tabularnewline
	$\kappa(\text{ref}(r))$				&
	$\text{ref}(\rho(r,\geqslant 1))$
	\tabularnewline
	$\kappa(\text{irr}(r))$		&
	$\text{irr}(\rho(r,> 0))$
	\tabularnewline
	$\kappa(\text{sym}(r))$ &
	$\kappa(r\sqsubseteq r^-)$
	\tabularnewline
	$\kappa(\text{asy}(r))$	&	
	$\text{asy}(\rho(r,> 0))$
	\tabularnewline
\bottomrule
\tabularnewline
\end{tabular}
\end{table*}
}

Table~\ref{table: SROIQ reduction} depicts the reduction rules for transforming
a \chain-\SROIN ontology~\Onto into a classical \SROIN ontology $\Onto\crisp$.
In this table, we use the notation
\[ {\negation^-}(d) := \max\{d'\in\chain \mid \negation d'\geqslant d\}. \]
Likewise, we define ${\tnorm^-}(d)$ as the set of all pairs
$(d_1,d_2)\in\chain^2$ that satisfy $d_1\tnorm d_2\geqslant d$ and are minimal
w.r.t.\ the component-wise ordering on $\chain^2$. This means that all elements
of ${\tnorm^-}(d)$ are incomparable, i.e.\ for all
$(d_1,d_2),(d_1',d_2')\in{\tnorm^-}(d)$ we have either $d_1>d_1'$ and
$d_2<d_2'$ or vice versa.
The set ${\tconorm^-}(d)$ is defined analogously.
For the implication, we need a slightly different definition, characterizing
all pairs of elements whose implication does \emph{not} exceed a specified
value~$d$. More precisely, we define ${\implication^-}(d)$ as the set of all
$(d_1,d_2)\in\chain^2$ satisfying $d_1\implication d_2<d$, and minimize here
w.r.t.\ the first component and maximize w.r.t.\ the second component since
\implication is antitone in the first argument and monotone in the second
argument.

Note that all expressions of the form $>d$ in
Table~\ref{table: SROIQ reduction} are well-defined since we have $d<1$ in all
such cases.
In particular, it holds that ${\negation^-}(d)<1$ whenever $d>0$, and $d_2<1$
for all $(d_1,d_2)\in{\implication^-}(d)$.

We can now prove Theorem~\ref{thm:crispification-correct}. In the following, let
\Onto be an arbitrary (not necessarily normalized) ontology in \chain-\SROIN, 
and
$\Onto\crisp$ be its reduced form according to
Table~\ref{table: SROIQ reduction}.

\subsection{Soundness}
\label{app:soundness}

Consider first a classical model \J of~$\Onto\crisp$. We
construct a fuzzy interpretation \I, with the goal of showing that \I is a model
of~\Onto, as follows
(for all $x,y\in\Delta^{\J}$, $a\in\NI$, $A\in\NC$, $r\in\NR$, and
$d\in\chain$):
\begin{alignat*}{3}
\Delta^{\I} & :=\Delta^{\J} \\
a^\I&:=a^\J \\
A^{\I}\left(x\right) & :=\max\big\{ {d}\mid x\in A_{\geqslant 
{d}}^{\J}\big\} \\
r^{\I}\left(x,y\right) & :=\max\big\{ {d}\mid(x,y)\in 
r_{\geqslant {d}}^{\J}\big\}
\end{alignat*}
In order to show that $\I$ is a model of $\Onto$, we first prove the following
proposition:

\begin{proposition}\label{prop 1:I->J}
	Let $C$ be a concept, $r$ a role, $x,y\in \Delta^\I$, and
	$d\in\chain_>0$. Then we have
	\begin{alignat*}{3}
		{C}^{\I}\left(x\right)\geqslant {d} &\text{ iff }
		  x\in \rho\left(C,\geqslant {d}\right)^\J \text{ and}
		\\
		{r}^{\I}\left(x,y\right)\geqslant{d} &\text{ iff }
		  (x,y)\in\rho\left(r,\geqslant{d}\right)^{\J}.
	\end{alignat*}
\end{proposition}

\begin{proof}
	For role names~$r$, the claim holds by the construction of $r^\I$ and the fact
	that $r_{\geqslant\next{d}}\sqsubseteq r_{\geqslant d}$ is contained in
	$\Onto\crisp$ for every $d\in\chain_{<1}$.
	For the universal role~$u$, the equivalence trivially holds since we have
	that $(x,y)\in u^\J$ and $u^\I(x,y)=1$.	
	Finally, for inverse roles it is an immediate consequence of the facts
	that $(r^{-})^\I(x,y)=$ $r^\I(y,x)$, and
	$(x,y)\in (r_{\geqslant d}^-)^\J$ iff
	$(y,x)\in r_{\geqslant d}^\J$.

	The rest of the proposition is proved by induction on the structure of~$C$.
	For concept names~$A$, the equivalence holds by the definition of $A^\I$ and
	the fact that $A_{\geqslant\next{d}}\sqsubseteq A_{\geqslant  d}$ is in
	$\Onto\crisp$ for every $d\in\chain_{<1}$.
	For the induction step, we consider all possible concept constructors:

	\noindent\textbf{Top Concept:} 
	The proof for this case is an immediate consequence of the facts that
	$\top^\I(x)=1$, $\rho(\top,\geqslant d)=\top$, and
	$\top^\J=\Delta^\J=\Delta^\I$.

	\noindent\textbf{Bottom Concept:} 
	For bottom, we  have $\bot^\I(x)=0$,
	$\rho(\bot,\geqslant d)=\bot$, and $\bot^\J=\emptyset$.

	\noindent\textbf{Concept Negation:} 
	We have that $x \in \rho(\lnot C,\geqslant d)^\J$ holds if and only if
	$x \notin \rho(C,>{\negation^-}(d))^\J$. By the induction hypothesis,
	this is equivalent to
	$C^\I(x)\leqslant\max\{d'\in\chain \mid \negation d'\geqslant d\}$.
	Since $\negation$ is antitone, this is finally
	equivalent to $(\lnot C)^\I(x)=\negation C^\I(x)\geqslant d$.


	\noindent\textbf{Concept Conjunction:} 
	If $C_1^\I(x)\tnorm C_2^\I(x)\geqslant {d}$, then by the definition
	of ${\tnorm^-}(d)$ there is at least one pair $(d_1,d_2)\in{\tnorm^-}(d)$
	such that $C_1^\I(x)\geqslant d_1$ and $C_2^\I(x)\geqslant d_2$.
	Since $d>0$, we also know that $d_1>0$ and $d_2>0$.
	Thus, by the induction hypothesis we have
	\[
		x\in
		  \big(\rho(C_1,\geqslant {d_1})\sqcap \rho(C_2,\geqslant {d_2})\big)^\J
		  \subseteq \rho(C_1\sqcap C_2,\geqslant d)^\J.
	\]
	
	Conversely, suppose that  
	$x\in(\rho(C_1,\geqslant d_1)\sqcap\rho(C_2,\geqslant d_2))^\J$
	holds for some $(d_1,d_2)\in{\tnorm^-}(d)$.
	Then by induction hypothesis, $C_1^\I(x)\geqslant {d_1}$ and
	$C_2^\I(x)\geqslant {d_2}$.
	Because of the monotonicity of \tnorm we then get that
	$C_1^\I(x)\tnorm C_2^\I(x)\geqslant d_1\tnorm d_2\geqslant d$, as we
	wanted to show.

	The proof for disjunction is similar to the proof for conjunction.

	\noindent\textbf{Existential Restriction:} Suppose that
	  $(\exists r. C)^\I(x)\geqslant {d}$.
	Since \chain is finite, there must exist some 
	$y\in\Delta^\I$ with 
	$r^\I(x,y)\tnorm C^\I(y)\geqslant {d}$.
	We have $r^\I(x,y)\geqslant {d_1}$
	and $C^\I(y)\geqslant {d_2}$ for some $({d_1},{d_2})\in{\tnorm^-}(d)$.
	By induction, we get
	$(x,y)\in\rho(r,\geqslant {d_1})^\J$ and
	$y\in\rho(C,\geqslant {d_2})^\J$.
	Therefore,
	\[
		x
		  \in \big(\exists\rho(r,\geqslant {d_1}).\rho(C,\geqslant {d_2}) \big)^\J
		  \subseteq \rho(\exists r .C, \geqslant {d})^\J.
	\]

	Conversely, suppose that 
	$x\in (\exists \rho(r,\geqslant {d_1}).\rho(C,\geqslant {d_2}))^\J$ for
	some pair $({d_1},{d_2})\in{\tnorm^-}(d)$.
	Thus, there exists $y$ with $(x,y)\in 
	\rho(r,\geqslant {d_1})^\J$ and $y \in \rho(C,\geqslant {d_2})^\J$.
	By induction, we have that
	$r^\I(x,y)\geqslant d_1$ and $C^\I(y)\geqslant d_2$, and 
	therefore
	\[ (\exists r.C)^\I(x)
	  \geqslant r^\I(x,y)\tnorm C^\I(y)
	  \geqslant d_1\tnorm d_2
	  \geqslant d. \]
	

	\noindent\textbf{Universal Restriction:}
	If $(\forall r.C)^\I(x)\geqslant d$, then for all
	$y\in\Delta^\I$ we have
	$\big(r^\I(x,y)\Rightarrow C^\I(y)\big)\geqslant {d}$.
	Consider any $y\in\Delta^\I$ and $(d_1,d_2)\in{\implication^-}(d)$
	with $r^\I(x,y)\geqslant d_1$. If $C^\I(y)\leqslant d_2$, then
	we immediately get that 
	$$r^\Imc(x,y)\implication C^\I(y)  \leqslant d_1\implication d_2 < d,$$
	which contradicts our assumption.
	Thus, we must have $C^\I(y)>d_2$, and hence
	$x\in\big(\forall\rho(r,\geqslant d_1).\rho(C,>d_2)\big)^\J$ by the
	induction hypothesis (as $d_1>0$).
	As this argument applies to all pairs $(d_1,d_2)\in{\implication^-}(d)$,
	we obtain $x\in\rho(\forall r.C,\geqslant d)^\J$, as required.
	
	For the opposite direction, assume that $(\forall r.C)^\I(x)<d$. Then
	there must be a $y\in\Delta^\I$ such that
	$r^\I(x,y)\implication C^\I(y)<d$. By the definition of
	${\implication^-}(d)$, we can find a pair $(d_1,d_2)\in{\implication^-}(d)$
	with $r^\I(x,y)\geqslant d_1$ and $C^\I(y)\leqslant d_2$. Since
	$d_1>0$, the induction hypothesis yields that 
	$(x,y)\in\rho(r,\geqslant d_1)^\J$ and
	$y\notin\rho(C,>d_2)^\J$.
	But then, this implies that $x\notin\rho(\forall r.C,\geqslant d)^\J$.
	
	\noindent\textbf{Nominals:}
	Consider the case where $C=\{d_1/o_1,\ldots,d_m/o_m\}$ such that
	$o_1,\dots,o_m\in\NI$ and $d_1,\ldots d_m\in\chainp$.
	Then $C^\I(x)\geqslant {d}$ iff $x=o_i^\I$ for some 
	$i\in\{1,\dots,m\}$ with $d_i\geqslant {d}$, which in turn is equivalent to
	\[ x
	  \in \{o_i^\I \mid d_i\geqslant d,\ i\in\{1,\dots,m\}\}
	  = \rho(C,\geqslant {d})^\J.
	\]
	
		%

  \noindent\textbf{Unqualified Number Restrictions:}
  The fact that $(\uatLeast{m}{r})^\I(x)\geqslant d$ is equivalent to the
  existence of $m$ different elements $y_1,\dots,y_m\in\Delta^\I$ such that
  $r^\I(x,y_i)\geqslant d$ holds for all $i\in\{1,\dots,m\}$. This is in turn
  equivalent to the existence of such $y_i$ with
  $(x,y_i)\in\rho(r,\geqslant d)^\J$ for all~$i$, and hence to
  $x\in(\uatLeast{m}{\rho(r,\geqslant d)})^\J$.

  The proof for unqualified at-most restrictions can be obtained by a
  combination of previous arguments for $\uatLeast{m}{r}$ and $\lnot C$.

	\noindent\textbf{Local Reflexivity:} 
	We have $(\exists r.\text{Self})^\I(x)\geqslant {d}$ iff 
	$r^\I(x,x)\geqslant {d}$, which is equivalent to
	$(x,x)\in \rho(r,\geqslant {d})^\J$, and to 
	$x\in(\exists \rho(r,{\geqslant {d}}).\text{Self})^\J$.
	\qed
\end{proof}
%
In order to finish the proof of the first direction of
Theorem~\ref{thm:crispification-correct}, it remains to show the following
lemma.

\begin{lemma}\label{lem: soundness}
  If \J is a classical model of $\Onto\crisp$, then \I is a fuzzy model
  of~\Onto.
\end{lemma}
\begin{proof}
  We need to show that \I satisfies all axioms in~$\Onto$:

	\noindent\textbf{Concept Assertions:}
	Suppose that $\Onto$ contains the concept assertion $C(a)\geqslant d$, and
	thus $\Onto\crisp$ contains the $\rho(C,\geqslant {d})(a)$.
	Since $\J$ is a model of $\Onto\crisp$, we have
	$a^\J\in\rho(C,\geqslant d)^\J$, and by Proposition~\ref{prop 1:I->J} that
	$C^\I(a^\I)\geqslant d$.
	Similarly, for an assertion $C(a)\leqslant d$ in~\Onto, we have
	$a^\J\notin\rho(C,>d)^\J$, and thus $C^\I(a^\I)\leqslant d$.

	\noindent\textbf{Other ABox Axioms:}
	The proof for role assertions can be obtained by adapting the proof for
	concept assertions.
	Axioms of the form $a\neq b$, $a=b$ are trivially satisfied since
	$\Delta^\I=\Delta^\J$ and for every individual $a\in\NI$ we have $a^\I=a^\J$.

	\noindent\textbf{Concept Inclusions:}
	Suppose that our ontology contains the concept inclusion
	$\axiom{C\sqsubseteq D \geqslant d}$ and assume that there is a
	$x\in\Delta^\I$ such that $C^\I(x)\implication D^\I(x)<d$. Thus,
	there exists a pair $(d_1,d_2)\in{\implication^-}(d)$ such that
	$C^\Imc(x)\geqslant d_1$ and $D^\Imc(x)\leqslant d_2$. Since $d_1>0$,
	Proposition~\ref{prop 1:I->J} yields $x\in\rho(C,\geqslant d_1)^\J$ and
	$x\notin\rho(D,>d_2)^\J$, which contradicts the fact that \J satisfies
	$\rho(C,\geqslant d_1)\sqsubseteq\rho(D,>d_2)$.

	\noindent\textbf{Role Inclusions:}
	Suppose that our ontology contains the role inclusion
	$\axiom{r_1r_2\sqsubseteq r\geqslant d}$ and it holds that
	\[ \big(r_1^\I(x,y)\otimes r_2^\I(y,z)\big)
	  \implication r^\I(x,z) < d, \]
	or equivalently,
	\[ r_1^\I(x,y)\implication
	  \big(r_2^\I(y,z)\implication r^\I(x,z)\big) < d, \]
	for some $x,y,z\in\Delta^\I$.
	Then there exist $(d_1,d')\in{\implication^-}(d)$ such that 
	$r_1^\I(x,y)\geqslant d_1$ and
	$r_2^\I(y,z)\implication r^\I(x,z)\leqslant d'$.
	The latter implies the existence of
	$(d_2,d_3)\in{\implication^-}(\next{d'})$ with
	$r_2^\I(y,z)\geqslant d_2$ and $r^\I(x,z)\leqslant d_3$.
	Proposition~\ref{prop 1:I->J} yields that
	$(x,y)\in\rho(r_1,\geqslant d_1)^\J$,
	$(y,z)\in\rho(r_2,\geqslant d_2)^\J$, and
	$(x,z)\notin\rho(r_3,>d_3)^\J$, which contradicts the fact that \J
	satisfies
	$\rho(r_1,\geqslant d_1)\rho(r_2,\geqslant d_2)\sqsubseteq\rho(r,>d_3)$.

	\noindent\textbf{Disjoint Role Axioms:}
	Suppose that our ontology contains the axiom $\text{dis}(r_1,r_2)$.
	We show that for all $x,y\in\Delta^\I$, either
	$r_1^\I(x,y)=0$ or $r_2^\I(x,y)=0$.
	Since \J satisfies $\Onto\crisp$, 
	$\rho(r_1,>0)^\J\cap\rho(r_2,>0)^\J=\emptyset$.
	By Proposition~\ref{prop 1:I->J}, there can be no pair
	$x,y\in\Delta^\I$ such that $r_1^\I(x,y)>0$ and
	$r_2^\I(x,y)>0$, as we wanted to show.

	The proofs for the other role axioms are similar.\qed
\end{proof}

\subsection{Completeness}
\label{app:completeness}

Conversely, we consider a fuzzy model \I of \Onto, and define the classical
interpretation \J as follows
(for all $x,y\in\Delta^{\I}$, $a\in\NI$, $A\in\NC$, $r\in\NR$, and
$d\in\chainp$):
\begin{alignat*}{3}
\Delta^{\J} & := \Delta^{\I} \\
a^\J\nonumber&:= a^\I\\
A_{\geqslant {d}}^{\J}&:=\{x\mid A^{\I}\left(x\right) \geqslant d \}	
\\
r_{\geqslant {d}}^{\J} &:= \{(x,y)\mid r^{\I}\left(x,y\right) 
\geqslant d\}
\end{alignat*}
We again prove a connection similar to the one of Proposition~\ref{prop 1:I->J}.

\begin{proposition}\label{prop: O -> Ocrisp}
  Let $C$ be a concept, $r$ a role, $x,y\in \Delta^\I$, and
  $d\in\chain_>0$. Then we have
  \begin{alignat*}{3}
  	x\in \rho\left(C,\geqslant {d}\right)^\J &\text{ iff }
  	{C}^{\I}\left(x\right)\geqslant {d} 
    \ \text{ and}
  	\\
  	(x,y)\in\rho\left(r,\geqslant{d}\right)^{\J}&\text{ iff }
   {r}^{\I}\left(x,y\right)\geqslant{d}.
  \end{alignat*}
\end{proposition}

\begin{proof}
  The proof is nearly the same as for Proposition~\ref{prop 1:I->J}, the only
  difference being the induction base cases.
	But it is easy to show the claim for concept and role names, given the
	definition of~\J.\qed
\end{proof}

\begin{lemma}\label{lem: completeness}
  If \I is a fuzzy model of \Onto, then \Jmc is a classical model
  of~$\Onto\crisp$.
\end{lemma}
\begin{proof}
  We need to show that \J satisfies all axioms in~$\Onto\crisp$:

	\noindent\textbf{Concept Assertions:}
	Suppose that $\Onto\crisp$ contains the concept assertion
	$\rho(C,\geqslant {d})(a)$. By the construction of $\Onto\crisp$,
	$C(a)\geqslant d$ appears in $\Onto$.
	Since $\I$ is a model of $\Onto$, we have $C^\I(a)\geqslant d$, and by 
	Proposition~\ref{prop: O -> Ocrisp} we get $a^\J\in\rho(C,\geqslant d)^\J$, as
	we wanted to show.

	If $\Onto\crisp$ contains an assertion $\neg \rho(C,> {d})(a)$,
	then $C(a)\leqslant d$ appears in $\Onto$ and consequently
	$C^\I(a)\leqslant d$. By Proposition~\ref{prop: O -> Ocrisp}, we have that
	$a\not\in \rho(C,>d)^\J$, as we wanted to show.

	\noindent\textbf{Other ABox Axioms:}
	The proof for role assertions can be obtained by adapting the proof for concept assertions.
	Axioms of the form $a\neq b$, $a=b$ are trivially satisfied since $\Delta^\I=\Delta^\J$ and for every individual $a\in\NI$ we have $a^\I=a^\J$.

	\noindent\textbf{Concept Inclusions:}		
	Suppose that our ontology contains a concept inclusion $\rho(C,\geqslant d_1)\sqsubseteq \rho(D,>d_2)$ that is not satisfied.
	Thus  there exists some $x\in\Delta^\J$ such that $x\in
	\rho(C,\geqslant d_1)^\J$ and $x\not\in \rho(D,> d_2)^\J$ and, by
	Proposition~\ref{prop: O -> Ocrisp},  $C^\I(x)\geqslant d_1$ and
	$D^\I(x)\leqslant d_2$.
	By the construction of $\Onto\crisp$, we have
	$\axiom{C\sqsubseteq D \geqslant d}$ in \Onto and $(d_1\Rightarrow d_2)< d$
	for some $d\in\chainp$.
	By the properties of~\implication, we get
	$C^\I(x)\Rightarrow D^\I(x)\leqslant d_1\Rightarrow d_2<d$, which
	contradicts our assumption that \I is a model of~\Onto.

	All concept inclusions of the form $A_{\geqslant\next{d}}\sqsubseteq 
	A_{\geqslant  d}$ are trivially satisfied by the construction of $\J$.
	
	\noindent\textbf{Role Inclusions:}
	Assume that a role inclusion
	\[ \rho(r_1,\geqslant d_1)\rho(r_2,\geqslant d_2)\sqsubseteq\rho(r,>d_3)
	  \in \Onto\crisp \]
	is violated, i.e.\ there are three elements $x,y,z\in\Delta^\J$ such that we
	have
	$(x,y)\in \rho(r_1,\geqslant d_1)^\J$,
	$(y,z)\in\rho(r_2,\geqslant d_2)^\J$, and
	$(x,z)\not\in\rho(r,> d_3)^\J$.
	Proposition~\ref{prop: O -> Ocrisp} implies that:
	\begin{alignat}{3}
		r_1^\I(x,y)&\geqslant d_1\text{,\qquad}&
		r_2^\I(y,z)&\geqslant d_2\text{,\qquad}&
		r(x,z)&\leqslant d_3.
		\label{eq: 2 O -> Ocrisp}
	\end{alignat}
	By construction of $\Onto\crisp$, we have that
	$\axiom{r_1r_2\sqsubseteq r\geqslant d}\in\Onto$,
	$(d_1\implication d')<d$, and $(d_2\implication d_3)<\next{d'}$ for some
	$d,d'\in\chainp$.
	We obtain
	$$\big(d_1\implication(d_2\implication d_3)\big)
	  \leqslant (d_1\implication d') < d,$$ 
	  and hence 
	$\big((d_1\tnorm d_2)\implication d_3\big) < d$.

	Along with~\eqref{eq: 2 O -> Ocrisp} and the monotonicity and
	antitonicity properties of the operators $\tnorm$ and $\Rightarrow$, this
	implies that
	\[\left(\left(r_1^\I(x,y)\tnorm 
	r_2^\I(y,z)\right)\Rightarrow r^\I(x,z)\right)<d\]
	which is absurd since $\I$ is a model of $\Onto$ and $\axiom{r_1r_2\sqsubseteq r\geqslant d}\in\Onto$.
	
	All role inclusions of the form
	$r_{\geqslant\next{d}}\sqsubseteq r_{\geqslant  d}$ are trivially satisfied
	by the construction of $\J$.

	\noindent\textbf{Disjoint Role Axioms:}
	Suppose that $\Onto\crisp$ contains the disjoint role  axiom $\text{dis}(\rho(r_1,>0),\rho(r_2,>0))$.
	By construction of $\Onto\crisp$, we also have that
	$\text{dis}(r_1,r_2)\in\Onto$, and therefore either $r_1(x,y)=0$ or
	$r_2(x,y)=0$  for all $x,y\in\Delta^\I$.
	Proposition~\ref{prop: O -> Ocrisp} now implies that the axiom 
	$$\text{dis}(\rho(r_1,>0),\rho(r_2,>0))$$ is satisfied.

	The proofs for the other role axioms are similar.\qed
	\end{proof}

\section{Proof of Lemma~\ref{lem:crispification-complexity}}
\label{sec:ap-lem13}

To determine the size of $\Onto\crisp$ for a normalized \chain-\SROIN
ontology~\Onto, we start by analyzing the size of the sets ${\tnorm^-}(d)$,
${\tconorm^-}(d)$ and ${\implication^-}(d)$ (defined in the beginning of Appendix~\ref{app:reduction-correct}). It is clear that for every
$d_1\in\chain$ there can be at most one element $d_2\in\chain$ such that
$(d_1,d_2)$ is contained in any of these sets. This is due to the minimization
conditions in their definitions. Thus, the size of these sets is at most linear
in the size of~\chain.
Consequently, the size of any expression of the form $\rho(C,\geqslant d)$,
where $C$ is a complex concept that contains only one concept constructor, is
at most linear in the sizes of~$C$ and~\chain (cf.\ 
Table~\ref{table: SROIQ reduction}).

Since \Onto is normalized, ABox axioms contain no complex concepts, and hence
the size of $\kappa(\alpha)$ for any such axiom~$\alpha$ is the same as the size of~$\alpha$.

Consider now a GCI $\alpha:=\axiom{C\sqsubseteq D\geqslant d}$ and its reduced
form, containing a GCI $\rho(C,\geqslant d_1)\sqsubseteq\rho(D,>d_2)$ for each
pair $(d_1,d_2)\in{\implication^-}(d)$. Since~$\alpha$ contains at most one
concept constructor, the size of each reduced axiom is linear in the sizes
of~$\alpha$ and~\chain. Moreover, there are at most linearly many such axioms
(in the size of~\chain), bringing the total size of $\kappa(\alpha)$ to at most
linear in the size of~$\alpha$ and quadratic in the size of~\chain.

Likewise, for a role inclusion $\alpha:=\axiom{r_1r_2\sqsubseteq r\geqslant d}$
the number of pairs $(d_1,d')\in{\implication^-}(d)$ is linear in the size
of~\chain, and for each of these pairs we additionally have to consider
linearly many pairs of the form $(d_2,d_3)\in{\implication^-}(\next{d'})$. Thus, the same
bounds are valid for role inclusions.
The proof for the remaining role axioms is trivial.

In summary, the total size of~$\Onto\crisp$ is bounded linearly in the size
of~\Onto and quadratically in the size of~\chain.
\qed